\documentclass[sigconf,nonacm]{acmart}
\usepackage{booktabs}
\usepackage{array}
\usepackage{amsmath}
\usepackage{xcolor}
\usepackage{tikz}
\usetikzlibrary{positioning,arrows.meta,decorations.pathreplacing,shapes.geometric}
\usetikzlibrary{shapes.callouts,shadows,calc,backgrounds,fit}
\usepackage{fontawesome5}
\usepackage{pgfplots}
\pgfplotsset{compat=1.18}
\definecolor{cMem}{HTML}{0072B2}
\definecolor{cRec}{HTML}{E69F00}
\definecolor{cOurs}{HTML}{009E73}
\definecolor{cEvid}{HTML}{D55E00}
\definecolor{cLLM}{HTML}{56B4E9}
\newcommand{\gain}[1]{\textbf{#1}}
\newcommand{\method}{\textsc{AgentMemRec}}

\newcommand{\Hay}{\mathcal{H}}
\newcommand{\Pool}{\mathcal{P}}
\newcommand{\Slate}{\mathcal{S}}
\newcommand{\Gold}{\mathcal{G}}
\newcommand{\tok}{\operatorname{tok}}
\newcommand{\rend}{\operatorname{render}}

\newcommand{\Acc}{\mathrm{Acc}}

\newcommand{\Kj}{K}
\newcommand{\rpool}{r_{\mathrm{pool}}}
\newcommand{\rrank}{r_{\mathrm{rank}}}
\newcommand{\rslate}{r_{\mathrm{slate}}}
\newcommand{\acomp}{a_{\mathrm{c}}}
\newcommand{\ainc}{a_{\mathrm{i}}}

\begin{document}

\title{Agent Memory as Recommendation: Recommender-System Design for Long-Term Memory of LLM Agents}
\author{Xinyu Zhang}
\affiliation{\institution{Fudan University}\country{}}
\email{zhangxinyu@fudan.edu.cn}
\author{Weiwei Sun}
\affiliation{\institution{Fudan University}\country{}}
\email{wwsun@fudan.edu.cn}

\begin{abstract}
Large language model (LLM) agents that serve users across many sessions must choose, for every request, the few stored memories that enter a limited context window.
This choice largely decides whether the agent answers correctly, and it runs on every request at Web scale.
This selection faces three challenges: a large and growing store, a per-request cost budget, and multi-session questions that need all of their evidence.
Many memory systems build notes or graphs with hundreds to thousands of LLM calls per history, and each is designed on its own, so it is unclear which component yields the accuracy.
These challenges mirror the properties of recommendation, so we formalize agent memory as budgeted slate recommendation, with the request as the user, stored memories as items and the context window as a slate.
The formalization places memory systems in one design space of recommender cascades and factorizes accuracy into the completeness that each stage controls.
Guided by it, we propose \method{}, which adapts candidate generation, learning to rank and slate construction from recommender systems and selects memories without any LLM call.
On multi-session questions of LongMemEval, \method{} outperforms memory and context-selection methods, including A-Mem, HippoRAG~2, DynamicRAG and T-Mem, by 10 to 31 points of accuracy.
It matches an LLM reranker that makes 64 calls per query at 3.2 times lower serving cost, and on LoCoMo it outperforms Mem0 and Zep without any LLM call at write time.
The factorization attributes the gains to slates that more often hold all of the evidence, so recommender-system design offers an effective and efficient foundation for the long-term memory of LLM agents.

\end{abstract}

\maketitle
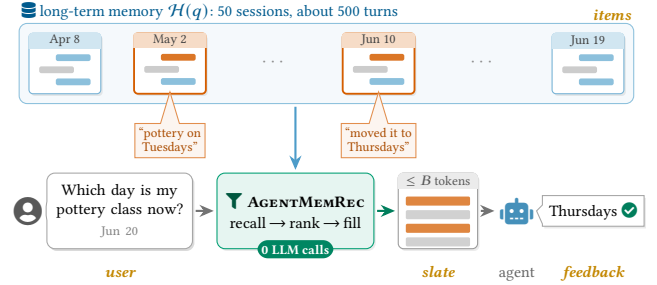
\begin{figure}[t]
  \centering
  \begin{tikzpicture}[>=Stealth, font=\footnotesize,
      shadowed/.style={drop shadow={opacity=0.16, shadow xshift=0.25mm, shadow yshift=-0.3mm}},
      card/.style={draw=cMem!55, fill=white, rounded corners=1.5pt, minimum width=0.95cm, minimum height=0.8cm, inner sep=0pt, shadowed},
      ecard/.style={card, draw=cEvid, line width=0.7pt},
      ub/.style={fill=cMem!45, rounded corners=0.8pt, minimum width=0.46cm, minimum height=0.09cm, inner sep=0pt},
      ab/.style={ub, fill=black!20},
      eb/.style={ub, fill=cEvid!85},
      box/.style={minimum height=1.0cm, rounded corners=3pt, shadowed},
      tag/.style={font=\scriptsize\itshape\bfseries, text=cRec!85!black, anchor=base},
      turn/.style={minimum width=0.85cm, minimum height=0.12cm, inner sep=0pt, rounded corners=0.6pt},
      say/.style={rectangle callout, draw=black!35, fill=white, align=center, inner sep=2.5pt, font=\scriptsize, shadowed},
      quote/.style={rectangle callout, draw=cEvid!70, fill=cEvid!6, align=center, inner sep=1.5pt, font=\tiny, text=cEvid!80!black},
      flow/.style={-{Stealth[length=1.9mm, width=1.6mm]}, thick, black!55, shorten <=1pt, shorten >=1pt}]
    \foreach \n/\x/\d/\e in {s1/0.69/Apr 8/0, s3/2.08/May 2/1, s17/4.84/Jun 10/1, s50/7.61/Jun 19/0} {
      \ifnum\e=1 \node[ecard] (\n) at (\x,2.95) {}; \else \node[card] (\n) at (\x,2.95) {}; \fi
      \begin{scope}
        \clip[rounded corners=1.5pt] (\n.north west) rectangle (\n.south east);
        \ifnum\e=1 \fill[cEvid!22] (\n.north west) rectangle ([yshift=-0.21cm]\n.north east);
        \else \fill[cMem!18] (\n.north west) rectangle ([yshift=-0.21cm]\n.north east); \fi
      \end{scope}
      \node[font=\tiny, text=black!75] at ([yshift=-0.105cm]\n.north) {\d};
      \ifnum\e=1 \node[eb] at ([xshift=0.11cm,yshift=-0.33cm]\n.north) {}; \else \node[ub] at ([xshift=0.11cm,yshift=-0.33cm]\n.north) {}; \fi
      \node[ab] at ([xshift=-0.11cm,yshift=-0.49cm]\n.north) {};
      \node[ub] at ([xshift=0.11cm,yshift=-0.65cm]\n.north) {};
    }
    \node[font=\scriptsize, text=black!45] at (3.46,2.95) {$\cdots$};
    \node[font=\scriptsize, text=black!45] at (6.22,2.95) {$\cdots$};
    \begin{scope}[on background layer]
      \node[draw=cMem!45, fill=cMem!4, rounded corners=3pt, fit=(s1)(s50), inner sep=3.5pt] (hist) {};
    \end{scope}
    \node[font=\scriptsize, text=cMem!80!black, anchor=south west, inner sep=1pt] at (hist.north west) {\faDatabase\ long-term memory $\Hay(q)$: 50 sessions, about 500 turns};
    \node[font=\scriptsize\itshape\bfseries, text=cRec!85!black, anchor=south east, inner sep=1pt] at (hist.north east) {items};
    \node[quote, callout absolute pointer={([yshift=1pt]s3.south)}] at (2.08,1.93) {``pottery on\\[-1pt]Tuesdays''};
    \node[quote, callout absolute pointer={([yshift=1pt]s17.south)}] at (4.84,1.93) {``moved it to\\[-1pt]Thursdays''};
    \node[font=\Large, text=black!65] (user) at (0.2,1.0) {\faUserCircle};
    \node[say, box, text width=1.75cm, callout absolute pointer={(0.43,1.0)}] (q) at (1.415,1.0) {Which day is my pottery class now?\\[1pt]\textcolor{black!50}{\tiny Jun 20}};
    \node[box, draw=cOurs, fill=cOurs!10, align=center, inner sep=2.5pt, text width=1.9cm, font=\scriptsize] (sel) at (3.74,1.0)
      {\textcolor{cOurs!55!black}{\faFilter}\ \textbf{\method{}}\\[1pt]recall$\,\to\,$rank$\,\to\,$fill};
    \node[fill=cOurs!85!black, text=white, rounded corners=4pt, font=\tiny\bfseries, inner sep=1.6pt] at (sel.south) {0 LLM calls};
    \node[box, draw=black!40, fill=white, rounded corners=2pt, minimum width=1.05cm, inner sep=0pt] (slate) at (5.625,1.0) {};
    \begin{scope}
      \clip[rounded corners=2pt] (slate.north west) rectangle (slate.south east);
      \fill[black!8] (slate.north west) rectangle ([yshift=-0.22cm]slate.north east);
    \end{scope}
    \node[font=\tiny, text=black!65] at ([yshift=-0.11cm]slate.north) {$\le B$ tokens};
    \foreach \y/\c in {1.13/cEvid!75, 0.95/black!18, 0.77/cEvid!75, 0.59/black!18} \node[turn, fill=\c] at (5.625,\y) {};
    \node[font=\large, text=cLLM!80!black, inner sep=1pt] (llm) at (6.67,1.0) {\faRobot};
    \node[say, anchor=west, callout absolute pointer={(6.93,1.0)}] (ans) at (7.0,1.0) {Thursdays \textcolor{cOurs!80!black}{\faCheckCircle}};
    \draw[flow] (q.east) -- (sel.west);
    \draw[flow, cOurs!80!black] (sel.east) -- (slate.west);
    \draw[flow] (slate.east) -- (llm.west);
    \draw[flow, cMem!70] (hist.south -| sel.north) -- (sel.north);
    \node[tag] at (1.415,0.12) {user};
    \node[tag] at (5.625,0.12) {slate};
    \node[font=\scriptsize, text=black!60, anchor=base] at (6.67,0.12) {agent};
    \node[tag] at (ans.south |- 0,0.12) {feedback};
  \end{tikzpicture}
  \caption{Agent memory as recommendation. Evidence for the request lies in two of 50 sessions, and a selector with no LLM call fills a slate of stored turns under a token budget $B$, from which the agent's LLM answers. Orange labels give the recommender role of each element.}
  \Description{A long-term memory of 50 dated sessions is drawn as a row of chat cards with user and assistant turns. Sessions 3 and 17 are highlighted because they hold the evidence, pottery on Tuesdays and moved to Thursdays. A user asks on June 20 which day the pottery class is now. AgentMemRec, which recalls, ranks and fills with no LLM call, selects a slate of at most B tokens that contains both evidence turns. An answering LLM reads the slate and answers Thursdays, which is judged correct. Orange labels mark the history as items, the request as the user, the context as the slate and the judged answer as feedback.}
  \label{fig:teaser}
\end{figure}

\section{Introduction}
\label{sec:intro}

Web-based chat assistants keep a long-term memory of each user's past sessions \citep{wu2025longmemeval}, so for every request their large language model (LLM) agent must decide which stored memories enter its limited context (Figure~\ref{fig:teaser}).
This decision largely determines whether the agent answers correctly, since a context that misses an evidence turn loses 37 to 64 points of accuracy with Qwen for every turn-level design we compare.
The decision also runs on every request, so its latency and cost grow with Web traffic, as they do for Web recommenders, and production memory layers are judged by these costs \citep{chhikara2025mem0}.

Three properties make the decision hard.
The store is large and keeps growing: in LongMemEval-M, the history of one question holds about 500 sessions and 1.5M tokens \citep{wu2025longmemeval}, while the answering model in our setting reads at most 2,700 tokens.
The budget is per request, so scoring the store with an LLM costs calls in proportion to its size and to traffic.
The value is judged on the whole context, since a question whose evidence spans several sessions needs all of it, and feedback arrives for the context and never for a single memory.

Memory systems build different structures, such as linked notes, knowledge graphs, temporal hierarchies and fact-augmented keys \citep{xu2025amem,gutierrez2025hipporag2,li2026timem,wu2025longmemeval}.
Many spend LLM calls to build them, about 440 to 2,270 per question for HippoRAG~2, A-Mem and T-Mem on LongMemEval-S.
Without a shared formulation of the decision, it is unclear which component produces the accuracy and what it should cost.
Existing selectors also leave evidence out: the four compared methods whose contexts are turns hold all of the evidence for only 0.50 to 0.68 of the multi-session questions.

We observe that the three properties are those of recommendation, and we formalize the decision as budgeted slate recommendation.
Recommender systems choose a few items from a large catalogue for every request under a latency budget, and they judge the shown set as a whole \citep{covington2016youtube,ie2019slateq}.
In memory selection, the request plays the user, stored turns are the items, and the context is a slate under a token budget.
The formalization defines a design space in which every memory system is a point, a cascade of memory construction, recall, ranking and context assembly.
Three results about this space reduce design to one problem: within an LLM budget, raise the product of pool completeness and ranking completeness, and fill the slate.

We build \method{} by solving this problem with no LLM call.
A wide pool from two cheap recall channels raises pool completeness, a ranker learned over ten cheap features raises ranking completeness, and the pool size is chosen by their product.
On the LongMemEval-S questions whose evidence spans several sessions, \method{} outperforms Adaptive-k, DynamicRAG, HippoRAG~2, A-Mem and T-Mem by 10 to 31 points of accuracy, and the memory and context-selection methods run on LongMemEval-M by 11 to 33 points.
It matches or exceeds a reranker that makes 64 LLM calls per question, at 3.2 times lower serving cost per request.
Under the shared LongMemEval-S protocol, it scores 82.5 on the test questions and 83.2 on all 500, above what MemOS and TiMem report.
On LoCoMo, the frozen cascade outperforms Mem0 and Zep with no LLM call at write time.
The factorization explains these gains: the slates of \method{} hold all of the evidence for 0.94 of the questions, and on jointly complete slates the systems do not differ significantly.

Our contributions are:
\begin{itemize}
  \item A formalization of long-term agent memory as budgeted slate recommendation, with a design space of memory systems and three results that reduce design to raising pool and ranking completeness within an LLM budget.
  \item \method{}, which solves this problem with no LLM call, outperforms LLM-built memories and context selectors on LongMemEval and production memory layers on LoCoMo, and matches a 64-call LLM reranker at 3.2 times lower cost.
  \item A measured factorization showing that slate completeness carries the gains, and a test of where item-wise ranking stops sufficing, at $\kappa\ge3$ and with near-miss filler.
\end{itemize}

\section{Related Work}
\label{sec:related}

\paragraph{Long-term memory for LLM agents.}
LoCoMo \citep{maharana2024locomo} and LongMemEval \citep{wu2025longmemeval} evaluate assistants on long multi-session histories.
Most memory systems decide at write time what to store and how to index it, as linked notes \citep{xu2025amem}, consolidated facts \citep{chhikara2025mem0}, temporal knowledge graphs \citep{rasmussen2025zep} or fact graphs searched by personalized PageRank \citep{gutierrez2025hipporag2}.
Others build hypergraphs, temporal hierarchies, boundary-guided segments or antecedent links \citep{yue2026hypermem,li2026timem,zhong2026hingemem,tan2026cable}, and verbatim chunks have been reported to outperform extracted artifacts on long conversations \citep{an2026fidelity}.
In the recommender view, these designs are item side-information or recall structure, and most of them spend LLM calls at write time.
\method{} works at read time over raw turns with no LLM call, so its cost per request stays flat as stores grow, and it can use a memory builder's units as a recall channel.

\paragraph{Context selection under a budget.}
Adaptive-k cuts a ranked list at the largest gap in its similarity scores \citep{taguchi2025adaptivek}, and DynamicRAG trains a reranker that chooses both the order and the number of passages \citep{sun2025dynamicrag}.
SetR replaces ranking with set selection by an LLM \citep{lee2025setr}, which lies outside a budget of no LLM calls, and ranked-list truncation learns where to cut a list \citep{bahri2020choppy,meng2024rlt}.
High-scoring passages that do not answer the question hurt \citep{cuconasu2024noise}, as do retrieved hard negatives \citep{jin2025longcontext} and input length alone \citep{du2025context}.
Like these selectors, \method{} orders candidates and fills a budget. It keeps the budget fixed because a fill dominates every cutoff of a ranking when the reward does not fall as the fill extends it.
We add a label-controlled contrast of ranked and random filler and a noise-floor control for the per-question headroom that cutoff methods assume.

\paragraph{Recommender techniques and memory.}
Industrial recommenders split candidate generation from ranking \citep{covington2016youtube} and treat the shown set as the unit of decision \citep{ie2019slateq}, and industrial search arranges rankers of increasing cost in cascades \citep{liu2017cascade}.
Set-aware stages re-rank with attention over the list \citep{pei2019prm} or select diverse subsets with determinantal point processes \citep{chen2018dpp}.
LambdaRank with gradient-boosted trees is a standard learned ranker \citep{burges2010ranknet,wu2010adapting,ke2017lightgbm}, reciprocal rank fusion merges channels \citep{cormack2009rrf}, and LLM rankers can be distilled into cross-encoders \citep{schlatt2025rankdistillm}.
The probability ranking principle orders documents by relevance when each is judged on its own \citep{robertson1977prp}.
For agent memory, RMM trains a reranker online from the citations of the answering model \citep{tan2025rmm}, and recent work studies learned memory rerankers \citep{li2026memreranker,pan2026convmemory,feng2026retriever}, pipelines without LLM calls \citep{piehl2026memfit} and small-model context construction \citep{yu2026lazymem}.
MemRec builds agentic recommenders on a collaborative memory \citep{chen2026memrec}, and RecMem consolidates agent memory by recurrence \citep{dai2026recmem}.
We formalize agent memory as budgeted slate recommendation, extend the probability ranking principle to complete evidence sets, and derive and measure which stages of the design space carry accuracy.
Unlike RMM, \method{} needs no feedback from the answering model at use time, which our frozen protocol could not provide, and we did not run these rerankers.
MemRec applies memory to recommendation and RecMem decides when to consolidate, whereas we select a context.

\section{Agent Memory as Recommendation}
\label{sec:problem}

This section casts memory selection as budgeted slate recommendation, lays out the design space that this view opens, and proves three results about the space that guide the design of \method{}.

\subsection{Memory selection as slate recommendation}
\label{sec:formal}

\begin{definition}[Budgeted slate recommendation]
\label{def:bsr}
A budgeted slate recommendation problem has a distribution of requests $x$, a catalogue $I(x)$ for each request, a monotone set cost $c_x$ on subsets of $I(x)$, a budget $B$, and a slate reward $R(\Slate,x)\in[0,1]$ that is observed only for the slate shown.
A policy $\pi$ maps each request to a slate with $c_x(\pi(x))\le B$, and its value is $V(\pi)=\mathbb{E}_x[R(\pi(x),x)]$.
\end{definition}

A question $q$ arrives at date $t_q$ with a history $\Hay(q)$, the set of user and assistant turns in its dated sessions.
Memory selection is an instance of Definition~\ref{def:bsr}.
The request is $x=(q,t_q)$ together with the agent's state, and the catalogue is $I(x)=\Hay(q)$, whose units are turns and, optionally, records that cite turns.
The cost is the token count of the rendered slate, $c=\tok\circ\rend$, which grows as units are added.
The reward $u(\Slate,q)$ is the probability that the judge accepts the answer that an answering model generates from $\rend(\Slate)$, so the selector solves
\begin{equation}
  \max_{\pi}\ \mathbb{E}_{q}\bigl[u(\pi(q),q)\bigr]\quad\text{s.t.}\quad \tok\bigl(\rend(\pi(q))\bigr)\le B\ \text{for every } q .
  \label{eq:objective}
\end{equation}
Judge accuracy $\Acc(\pi)$ estimates this objective, with expectations over questions, perturbed budgets and decoding.
The instance has the three properties that shape recommender design.
\emph{Catalogue scale}: a LongMemEval-S history holds about 490 turns and a LongMemEval-M history about 4,900, while the slate admits 2,700 tokens.
\emph{Per-request budget}: the selector runs on every request, so scoring the catalogue with an LLM costs calls in proportion to its size and to traffic.
\emph{Slate-level value}: $u$ is a set function, a question whose evidence spans $\kappa\ge2$ annotated sessions needs all of it, and feedback arrives for the slate and never for a single unit, as in slate recommendation \citep{ie2019slateq}.

\begin{figure*}[t]
  \centering
  \begin{tikzpicture}[>=Stealth, font=\scriptsize,
      cell/.style={draw=black!25, fill=black!3, rounded corners=1.5pt, minimum height=0.36cm, inner sep=1.5pt, align=center},
      llm/.style={cell, draw=cLLM!80!black, fill=cLLM!15},
      ours/.style={cell, draw=cOurs!80!black, fill=cOurs!12, font=\scriptsize\bfseries},
      grp/.style={anchor=west, font=\scriptsize\itshape, text=black!55},
      sysname/.style={anchor=west},
      bub/.style={minimum width=0.24cm, minimum height=0.09cm, inner sep=0pt, rounded corners=0.8pt},
      bar/.style={minimum height=0.1cm, inner sep=0pt, rounded corners=0.4pt, fill=black!20},
      ebar/.style={bar, fill=cEvid!75},
      term/.style={font=\scriptsize\itshape\bfseries, text=cRec!85!black},
      obj/.style={font=\scriptsize, text=black!70},
      hd/.style={font=\scriptsize\bfseries, align=center}]
    \def\xP{4.45}\def\wP{2.9cm} \def\xG{7.7}\def\wG{3.1cm} \def\xF{10.95}\def\wF{2.9cm} \def\xH{13.85}\def\wH{2.4cm} \def\xS{16.45}\def\wS{2.0cm}
    \newcommand{\C}[4][cell]{\node[#1, text width=#3] at (#2,\y) {#4};}
    \newcommand{\R}{\,\textcolor{cLLM!70!black}{\faRobot}}
    \newcommand{\I}[1]{\textcolor{cMem!75!black}{#1}\ }
    \node[rotate=90, font=\scriptsize\bfseries, text=black!60] at (0.18,6.75) {(a) cascade};
    \foreach \x/\m/\r in {\xP/memory construction $\varphi$/write-time map, \xG/memory recall $g$/candidate generation, \xF/memory ranking $f$/ranker, \xH/context assembly $h$/slate builder, 16.45/answer and judge/feedback} {
      \node[font=\scriptsize\bfseries, text=black!80] at (\x,8.38) {\m};
      \node[term] at (\x,8.08) {\r};}
    \node[font=\normalsize, text=black!65] at (0.72,7.3) {\faUserCircle};
    \node[anchor=west, align=left, inner sep=0pt] at (0.98,7.3) {request $q$\\at date $t_q$};
    \node[font=\normalsize, text=cMem!85!black] at (0.72,6.25) {\faDatabase};
    \node[anchor=west, align=left, inner sep=0pt] at (0.98,6.25) {long-term\\memory $\Hay(q)$};
    \shade[left color=cMem!10, right color=cOurs!14, draw=black!25] (2.75,7.75) -- (15.15,7.25) -- (15.15,6.25) -- (2.75,5.75) -- cycle;
    \foreach \x/\h in {6.0/0.869, 9.375/0.733, 12.525/0.606} \draw[white, line width=1.3pt] (\x,6.75-\h) -- (\x,6.75+\h);
    \draw[->, thick, black!55] (2.25,7.3) -- (2.72,7.15);
    \draw[->, thick, black!55] (2.25,6.25) -- (2.72,6.35);
    \foreach \k in {0,...,5} {
      \pgfmathsetmacro{\cx}{3.12+0.47*\k}
      \node[draw=cMem!50, fill=white, rounded corners=1pt, minimum width=0.38cm, minimum height=0.86cm, inner sep=0pt] (c\k) at (\cx,6.68) {};
      \fill[cMem!25] ([xshift=0.25pt,yshift=-0.25pt]c\k.north west) rectangle ([xshift=-0.25pt,yshift=-0.14cm]c\k.north east);
      \foreach \j/\dx/\c in {0/0.04/cMem!45, 1/-0.04/black!20, 2/0.04/cMem!45, 3/-0.04/black!20}
        \node[bub, fill=\c] at (\cx+\dx,6.9-0.18*\j) {};}
    \node[bub, fill=cEvid!85] at (3.63,6.9) {}; \node[bub, fill=cEvid!85] at (5.04,6.54) {};
    \node[obj] at (\xP,7.4) {dated sessions of turns};
    \node[font=\tiny, text=cMem!75!black] at (\xP,6.03) {\faStickyNote\,\faProjectDiagram\,\faKey\ \ optional structure};
    \foreach \i in {0,1,2} \foreach \j/\c in {0/cMem!45, 1/black!20, 2/cMem!45, 3/black!20} \node[bub, minimum width=0.3cm, fill=\c] at (7.3+0.4*\i,6.95-0.2*\j) {};
    \node[bub, minimum width=0.3cm, fill=cEvid!85] at (7.3,6.75) {}; \node[bub, minimum width=0.3cm, fill=cEvid!85] at (8.1,6.35) {};
    \node[obj] at (\xG,7.25) {recalled pool $\Pool(q)$};
    \foreach \y/\s in {6.92/ebar, 6.75/ebar, 6.58/bar, 6.41/bar, 6.24/bar} \node[\s, minimum width=1.1cm] at (\xF,\y) {};
    \foreach \y/\r in {6.92/1, 6.75/2, 6.58/3} \node[font=\tiny, text=black!55, anchor=east] at (10.33,\y) {\r};
    \node[obj] at (\xF,7.17) {ranked memories $\sigma$};
    \node[draw=black!45, fill=white, rounded corners=1pt, minimum width=1.1cm, minimum height=0.66cm, inner sep=0pt] (ctx) at (13.85,6.62) {};
    \fill[black!10] ([xshift=0.25pt,yshift=-0.25pt]ctx.north west) rectangle ([xshift=-0.25pt,yshift=-0.13cm]ctx.north east);
    \foreach \y/\s in {6.73/ebar, 6.59/ebar, 6.45/bar} \node[\s, minimum width=0.9cm, minimum height=0.08cm] at (13.82,\y) {};
    \draw[cOurs!80!black, densely dashed, thick] (14.4,6.25) -- (14.4,6.99);
    \node[font=\tiny, text=cOurs!60!black, anchor=west, inner sep=1pt] at (14.42,6.62) {$B$};
    \node[obj] at (\xH,7.1) {agent context $\Slate_B$};
    \draw[->, thick, black!55] (15.18,6.75) -- (15.6,6.75);
    \node[font=\Large, text=cLLM!80!black] (llm) at (15.95,6.75) {\faRobot};
    \draw[->, thick, black!55] (16.3,6.75) -- (16.72,6.75);
    \node[font=\large] at (17.05,6.75) {\textcolor{cOurs!80!black}{\faCheckCircle}};
    \node[obj] at (16.45,7.25) {agent answers};
    \node[obj] at (16.45,6.25) {judged answer};
    \draw[decorate, decoration={brace, mirror, amplitude=3pt}, cRec!85!black] (2.8,5.66) -- (9.33,5.66) node[midway, below=2pt, term, align=center] {pool completeness $\rpool$\\[-1pt]\textmd{the pool keeps all evidence}};
    \draw[decorate, decoration={brace, mirror, amplitude=3pt}, cRec!85!black] (9.42,5.66) -- (12.48,5.66) node[midway, below=2pt, term, align=center] {ranking completeness $\rrank$\\[-1pt]\textmd{the slate keeps pooled evidence}};
    \draw[decorate, decoration={brace, mirror, amplitude=3pt}, cRec!85!black] (12.57,5.66) -- (15.13,5.66) node[midway, below=2pt, term, align=center] {budget fill\\[-1pt]\textmd{ranked turns that fit}};
    \draw[decorate, decoration={brace, mirror, amplitude=3pt}, cRec!85!black] (15.3,5.66) -- (17.5,5.66) node[midway, below=2pt, term, align=center] {answer accuracy\\[-1pt]\textmd{judged per slate}};
    \node[rotate=90, font=\scriptsize\bfseries, text=black!60] at (0.18,2.45) {(b) design space};
    \node[hd, anchor=west] at (0.5,4.82) {System};
    \node[hd] at (\xP,4.82) {memory construction $\varphi$};
    \node[hd] at (\xG,4.82) {memory recall $g$};
    \node[hd] at (\xF,4.82) {memory ranking $f$};
    \node[hd] at (\xH,4.82) {context assembly $h$};
    \node[hd] at (\xS,4.82) {LLM calls, by stage};
    \begin{scope}[on background layer]
      \fill[black!4, rounded corners=2pt] (15.32,1.06) rectangle (17.58,5.0);
    \end{scope}
    \draw[black!35, densely dashed] (15.24,1.06) -- (15.24,5.0);
    \draw[black!45] (0.5,4.62) -- (17.55,4.62);
    \node[grp] at (0.5,4.43) {LLM-built memories};
    \def\y{4.15} \node[sysname] at (0.55,\y) {LongMemEval design\,\citep{wu2025longmemeval}};
      \C[llm]{\xP}{\wP}{\I{\faKey}fact-augmented keys\R} \C{\xG}{\wG}{dense over keys} \C{\xF}{\wF}{similarity} \C{\xH}{\wH}{fill, whole rounds} \C[llm]{\xS}{\wS}{at $\varphi$\R}
    \def\y{3.73} \node[sysname] at (0.55,\y) {HippoRAG~2\,\citep{gutierrez2025hipporag2}};
      \C[llm]{\xP}{\wP}{\I{\faProjectDiagram}fact graph\R} \C[llm]{\xG}{\wG}{dense facts, LLM filter\R} \C{\xF}{\wF}{PageRank on the graph} \C{\xH}{\wH}{fill, whole rounds} \C[llm]{\xS}{\wS}{at $\varphi$ and $g$\R}
    \def\y{3.31} \node[sysname] at (0.55,\y) {A-Mem\,\citep{xu2025amem}};
      \C[llm]{\xP}{\wP}{\I{\faStickyNote}linked notes\R} \C[llm]{\xG}{\wG}{note retrieval\R} \C{\xF}{\wF}{similarity} \C{\xH}{\wH}{fill} \C[llm]{\xS}{\wS}{at $\varphi$ and $g$\R}
    \def\y{2.89} \node[sysname] at (0.55,\y) {T-Mem\,\citep{guo2026tmem}};
      \C[llm]{\xP}{\wP}{\I{\faLayerGroup}scenes, items, persona\R} \C{\xG}{\wG}{dense over scenes, items} \C{\xF}{\wF}{cross-encoder} \C{\xH}{\wH}{5 scenes, 15 items} \C[llm]{\xS}{\wS}{at $\varphi$\R}
    \node[grp] at (0.5,2.57) {context selectors};
    \def\y{2.29} \node[sysname] at (0.55,\y) {Adaptive-k\,\citep{taguchi2025adaptivek}};
      \C{\xP}{\wP}{\I{\faComments}turns} \C{\xG}{\wG}{dense over the history} \C{\xF}{\wF}{similarity} \C{\xH}{\wH}{largest-gap cutoff} \C{\xS}{\wS}{none}
    \def\y{1.87} \node[sysname] at (0.55,\y) {DynamicRAG\,\citep{sun2025dynamicrag}};
      \C{\xP}{\wP}{\I{\faComments}turns} \C{\xG}{\wG}{fused top 40} \C[llm]{\xF}{\wF}{LLM selector\R} \C{\xH}{\wH}{cutoff in the same call} \C[llm]{\xS}{\wS}{at $f$\R}
    \node[grp] at (0.5,1.55) {LLM reranker};
    \def\y{1.27} \node[sysname] at (0.55,\y) {Ref};
      \C{\xP}{\wP}{\I{\faComments}turns} \C{\xG}{\wG}{lexical + dense, fused, 64} \C[llm]{\xF}{\wF}{LLM scores, 64 calls\R} \C{\xH}{\wH}{fill} \C[llm]{\xS}{\wS}{at $f$: 64 calls\R}
    \draw[cOurs!70, fill=cOurs!5, rounded corners=2.5pt] (0.42,0.42) rectangle (17.6,0.98);
    \def\y{0.7} \node[anchor=west, font=\scriptsize\bfseries] at (0.5,\y) {\method{}};
      \C[ours]{\xP}{\wP}{\I{\faComments}turns, optional records} \C[ours]{\xG}{\wG}{lexical + dense, fused, 128} \C[ours]{\xF}{\wF}{learned, cheap features} \C[ours]{\xH}{\wH}{greedy fill to $B$} \C[ours]{\xS}{\wS}{none}
  \end{tikzpicture}
  \caption{Long-term agent memory as budgeted slate recommendation. (a)~Memory construction, recall, ranking and context assembly narrow the stored turns of a long-term memory to the agent's context under a token budget $B$, and each is a recommender stage (orange). (b)~The compared systems by their choice at each stage, where \textcolor{cLLM!70!black}{\faRobot} marks an LLM call and the last column summarizes the stages that call an LLM. \method{} learns its ranker over cheap features and makes no LLM call at any stage.}
  \Description{Two panels share four stage columns: write-time map, candidate generation, ranker and slate builder. Panel a shows a funnel that narrows from left to right. A request and a history enter it. The first segment holds the dated sessions of turns, among them two gold evidence units and some optional records. The second holds a smaller pool that still contains both gold units. The third holds a ranking with the gold units on top, and the fourth a slate under the token budget B. An answering LLM reads the slate, and the judged answer is the reward. Braces under the funnel state that the first two stages decide whether the pool keeps all evidence and that the ranker decides whether the slate keeps the pooled evidence. The slate builder fills the budget with the ranked turns that fit, and the judged answer gives the accuracy. Panel b lists the compared systems in three groups: LLM-built memories, context selectors and an LLM reranker. Each row has one chip per stage and a last column, set apart by a dashed line, that summarizes the stages at which the system calls an LLM, none for Adaptive-k and AgentMemRec. LongMemEval design, HippoRAG 2, A-Mem and T-Mem spend LLM calls on the write-time map, and HippoRAG 2 and A-Mem also at candidate generation. Adaptive-k and DynamicRAG use per-question cutoffs, and DynamicRAG and Ref call an LLM in the ranker. AgentMemRec uses turns with optional records, a fused lexical and dense pool of 128 units, a ranker learned over cheap features and a greedy fill, with no LLM call.}
  \label{fig:pipeline}
\end{figure*}
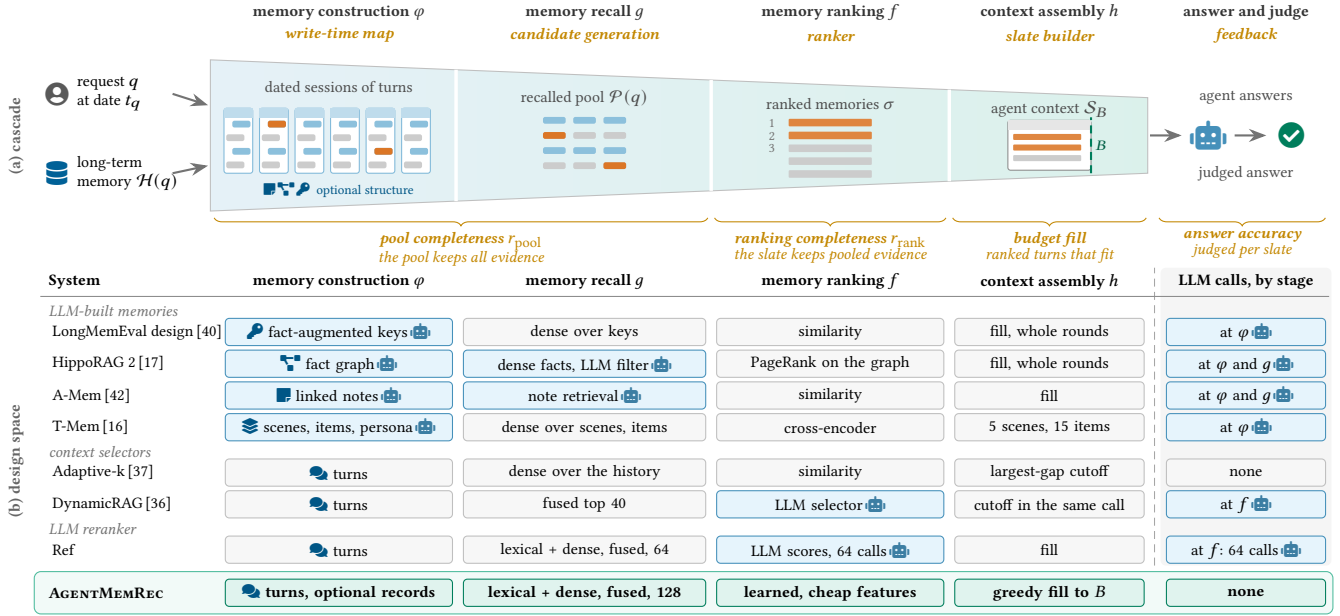

\subsection{The design space}
\label{sec:space}
Recommenders meet these three properties with a cascade, and we use the cascade to define the space of memory designs.

\begin{definition}[Design space]
\label{def:space}
A memory design is a tuple $d=(\varphi,g,f,h)$.
The write-time map $\varphi$ turns a history into a catalogue of units, possibly with side-information.
The candidate generator $g$ returns a pool $\Pool(q)\subseteq\varphi(\Hay(q))$, the ranker $f$ orders the pool, and the slate builder $h$ selects $\Slate_B(q)\subseteq\Pool(q)$ within the budget.
The LLM cost of $d$ counts the LLM calls of $\varphi$ per history and of $g$, $f$ and $h$ per request.
\end{definition}
Every memory system we compare is a point in this space (Figure~\ref{fig:pipeline}b), and the points differ mainly in where they spend LLM calls.

\subsection{Three results about the space}
\label{sec:props}
The benchmark marks a gold set $\Gold(q)\subseteq\Hay(q)$ of turns that contain the answer, and a slate is complete when the turns that its units render or cite cover $\Gold(q)$, written $\Gold(q)\subseteq\Slate_B(q)$.
For a policy, let $\rpool=\Pr[\Gold\subseteq\Pool]$ be its pool completeness, $\rrank=\Pr[\Gold\subseteq\Slate_B\mid\Gold\subseteq\Pool]$ its ranking completeness and $\rslate=\Pr[\Gold\subseteq\Slate_B]$ its slate completeness.
Let $\acomp$ and $\ainc$ be its mean reward on complete and on incomplete slates.

\begin{proposition}[Stage factorization]
\label{prop:factor}
For every cascade with $\Slate_B(q)\subseteq\Pool(q)$,
\(
  \Acc(\pi)=\ainc+(\acomp-\ainc)\,\rpool\,\rrank .
\)
\end{proposition}
\begin{proof}
By total expectation over the event $C=\{\Gold\subseteq\Slate_B\}$, $\Acc=\acomp\Pr[C]+\ainc(1-\Pr[C])$.
Since $\Slate_B\subseteq\Pool$, the event $C$ implies $\Gold\subseteq\Pool$, so $\Pr[C]=\rpool\,\rrank$.
\end{proof}
Proposition~\ref{prop:factor} separates the item-level stages from set-level effects.
Candidate generation sets $\rpool$, which bounds every ranker since $\rslate\le\rpool$, learning to rank sets $\rrank$, and the gap $\acomp-\ainc$ scales both gains.
Set-level effects, such as a misleading unit, act through $\acomp$ and $\ainc$, which also shift with which questions a policy completes and with how much gold an incomplete slate holds.

\begin{proposition}[Reduction to item-wise ranking]
\label{prop:itemwise}
Fix a request $q$, and suppose that the mean rewards on complete and on incomplete slates, $\acomp(q)>\ainc(q)$, do not depend on the slate, that unit costs are additive, and that under the selector's belief each unit $v$ belongs to $\Gold(q)$ independently with probability $0<p_v<1$.
Then maximizing the expected reward under the budget is a 0-1 knapsack with item values $-\log(1-p_v)$, and with equal unit costs the optimal slate holds the units with the largest $p_v$.
\end{proposition}
\begin{proof}
Applied given $q$, the total-expectation step of Proposition~\ref{prop:factor} gives the expected reward $\ainc(q)+(\acomp(q)-\ainc(q))P_{\Slate}$ with $P_{\Slate}=\Pr[\Gold\subseteq\Slate\mid q]=\prod_{v\notin\Slate}(1-p_v)$ by independence.
Conditioning on $\Gold\neq\emptyset$ maps $P_{\Slate}$ to $(P_{\Slate}-P_0)/(1-P_0)$ with $P_0=\prod_{v}(1-p_v)$, an increasing affine map that keeps the maximizer.
Taking logarithms, the slate maximizes $\sum_{v\in\Slate}-\log(1-p_v)$ under the budget, and with equal costs this sum is largest for the units with the largest $p_v$.
\end{proof}
Proposition~\ref{prop:itemwise} extends the probability ranking principle \citep{robertson1977prp} from single documents to complete evidence sets under a budget.
It is the reason a learned ranker of item relevance, followed by a fill, fits a set-level objective.
Rendered cost is monotone and submodular, since a session header is paid once, and with unequal costs the knapsack favours value per token, which the ranker sees through its length feature.
The proposition also names the conditions under which item-wise ranking stops being enough: rewards that depend on the composition of the slate, and evidence units whose relevance is correlated.

\begin{proposition}[Fill and cutoffs]
\label{prop:fill}
Let $\sigma$ rank the pool, and let $\Slate^{\sigma}_B$ be the greedy fill that walks $\sigma$ and adds each unit whose rendered cost still fits.
Every prefix $\sigma_{\le j}$ with $c(\sigma_{\le j})\le B$ satisfies $\sigma_{\le j}\subseteq\Slate^{\sigma}_B$.
Hence, if $u(\cdot,q)$ does not decrease on the sets between $\sigma_{\le j}$ and $\Slate^{\sigma}_B$, the fill is at least as good as the cutoff $\sigma_{\le j}$.
\end{proposition}
\begin{proof}
Let $i\le j$. By induction, $\sigma_1,\dots,\sigma_{i-1}$ are in the slate when the walk reaches $\sigma_i$.
The cost of $\sigma_{\le i}$ is at most $c(\sigma_{\le j})\le B$ because $c$ is monotone, so $\sigma_i$ is added.
The assumed monotonicity then gives $u(\sigma_{\le j},q)\le u(\Slate^{\sigma}_B,q)$.
\end{proof}
The fill need not be a prefix of $\sigma$, since it skips units that do not fit.
A per-question cutoff can therefore help only where some unit lowers the value of the slate, that is, through a negative externality.
Proposition~\ref{prop:fill} is the reason \method{} fills a fixed budget, and its condition is testable along learned rankings.

\subsection{The design problem}
\label{sec:design}
Together, the three results reduce the selection objective to a problem over the first three coordinates.
Suppose that the rewards on complete and incomplete slates do not depend on the slate and that evidence units are independent.
If the reward also does not fall as the fill extends a prefix, the slate builder should fill the budget.
A design should then solve
\begin{equation}
  \max_{\varphi,g,f}\ \rpool(\varphi,g)\;\rrank(f\mid\varphi,g)\quad\text{s.t.}\quad \mathrm{LLMcost}(d)\le L ,
  \label{eq:design}
\end{equation}
where $L$ is the LLM budget of the deployment.
Equation~\eqref{eq:design} joins the stages into one problem.
The pool couples the two factors, because a wider pool raises $\rpool$ and gives the ranker more distractors.
The constraint prices every coordinate, since an LLM call pays off only where it raises the product.
Equation~\eqref{eq:design} is a design objective, exact only under these conditions, and how well it orders designs is an empirical question.

Figure~\ref{fig:pipeline}b reads the compared systems against Equation~\eqref{eq:design}.
LLM-built memories spend hundreds to thousands of calls per history on $\varphi$, which the product credits only through the completeness that recall and ranking then reach.
Context selectors replace the fill with a per-question cutoff, which cannot beat the fill along a ranking whose reward is monotone, since the fill contains every prefix that fits.
The LLM reranker spends 64 calls per request on $f$, the coordinate that sets $\rrank$.
\method{} instantiates Equation~\eqref{eq:design} at $L=0$.
It raises $\rpool$ with a wide pool from two cheap channels and $\rrank$ with a ranker learned from relevance labels over cheap features, sizes the pool by the product, and fills the budget.
The recommender view thus prices every coordinate in LLM calls and predicts that per-question cutoffs cannot beat a fill when the reward does not fall along it. Its stage factorization locates the accuracy of any design in its stages.

\section{\method{}}
\label{sec:method}

\method{} instantiates the design problem of Equation~\eqref{eq:design} with an LLM budget of zero (Figure~\ref{fig:pipeline}).
Candidate generation raises $\rpool$, learning to rank raises $\rrank$, the pool size is chosen by their product, and the slate fills the budget instead of cutting the ranking.
No stage calls an LLM at selection time.

\subsection{Candidate generation: raising \texorpdfstring{$\rpool$}{pool completeness}}
\label{sec:recall}
In Equation~\eqref{eq:design}, $\rpool$ bounds the product, so the pool must be wide, and with $L=0$ it must be cheap to build.
Like a recommender, \method{} narrows the catalogue with cheap retrieval channels.
The lexical channel ranks all turns by BM25 \citep{robertson2009bm25} over lower-cased word tokens.
The dense channel ranks them by the cosine similarity of bge-m3 embeddings \citep{chen2024bgem3}.
Reciprocal rank fusion \citep[RRF;][]{cormack2009rrf} with $k=60$ merges the two lists, and the pool $\Pool(q)$ is the top $N=128$ fused turns.
The pool size was chosen on development questions by slate completeness after ranking, the product $\rpool\,\rrank$, where 128 turns led with 0.912 against 0.879 for 64 turns, in-sample for a ranker trained on 64-turn pools.

\subsection{Learning to rank: raising \texorpdfstring{$\rrank$}{ranking completeness}}
\label{sec:rank}
When evidence units are independent, the set-level objective reduces to ranking units by their probability of being evidence.
The ranker therefore learns item relevance from gold labels, and with $L=0$ it uses cheap features, as recommender cascades do.
\method{} uses a LightGBM model \citep{ke2017lightgbm} trained with the LambdaRank objective \citep{burges2010ranknet,wu2010adapting}.
Each (question, pooled turn) pair has ten features.
Five are cheap relevance signals: the bge-m3 cosine, the logarithms of the global dense and BM25 ranks, the position in the fused pool and the bge-reranker-v2-m3 cross-encoder score.
Five describe the turn: whether its role is user, its relative position in its session, the logarithms of its age in days and of its length, and the spaCy keys it shares with the question.
The labels are the gold annotations of the 152 development questions, with pools from both their LongMemEval-S and LongMemEval-M histories, 19,456 pairs in total.

\subsection{Slate construction: the fill}
\label{sec:slate}
The slate stage fills the budget $B$ from the ranked pool.
It walks the pool in rank order and adds each turn whose rendered cost still fits.
The rendered cost of a turn is the increase in the token count of the rendered slate when the turn is added, including its session header if no earlier turn of that session was selected.
A turn that does not fit is skipped, and the walk continues.
The rendering groups the selected turns by session, orders sessions by date and uses the official LongMemEval history format (Appendix~\ref{app:protocol}).
The budget is fixed at $B=2{,}700$ tokens, a value set in our first experiment before any system was compared, and the evaluation perturbs it to 2,400 to 2,700.
There is no per-question cutoff, because the fill contains every prefix that fits, so a cutoff can only help where a lower-ranked turn lowers the value of the slate.
On the ranking of \method{}, choosing the best of three budgets per question gains no more than decode noise.

\subsection{Optional write-time records}
\label{sec:m4}
This optional extension adds item side-information at write time to raise $\rpool$ and $\acomp$, with no LLM call at selection time.
An LLM (Qwen3-32B) reads each session once and writes typed records with a date, entities and the turns they cite.
A record channel adds to the 128-turn pool the turns cited by the records most similar to the question, up to 32 new turns, and the frozen ranker orders the union.
Relative-time tags attach the dates of citing records to turns that contain a relative time expression, keeping only dates up to the reference date of the question.

\section{Experimental Setup}
\label{sec:setup}
\begin{table*}[t]
  \centering\small
  \caption{Main results on the test questions of LongMemEval-S and LongMemEval-M with evidence in at least two sessions ($\kappa\ge2$, $n=209$ each). Gains are paired margins. LLM calls are per question, at write time (Write), amortized over the histories each build covers, and at selection time (Sel.). Tokens are mean slate lengths. Best in bold, second underlined.}
  \label{tab:s}
  \setlength{\tabcolsep}{5.3pt}
  \begin{tabular}[t]{lrrrcc}
    \toprule
    \multicolumn{6}{c}{\textbf{LongMemEval-S} (50 sessions, 100k tokens per history)} \\
    \midrule
    & \multicolumn{2}{c}{LLM calls $\downarrow$} & & \multicolumn{2}{c}{Accuracy (\%) $\uparrow$} \\
    \cmidrule(lr){2-3}\cmidrule(lr){5-6}
    System & Write & Sel. & Tokens & Qwen & DeepSeek \\
    \midrule
    Adaptive-k \citep{taguchi2025adaptivek} & 0 & 0 & 1,470 & 69.1 & 76.6 \\
    DynamicRAG-8B \citep{sun2025dynamicrag} & 0 & 1 & 2,140 & 67.0 & 74.2 \\
    HippoRAG~2 \citep{gutierrez2025hipporag2} & 442 & 1 & 2,520 & 59.9 & 63.2 \\
    A-Mem \citep{xu2025amem} & 979 & 1 & 2,540 & 54.4 & 58.1 \\
    T-Mem \citep{guo2026tmem} & 2,268 & 0 & 3,310 & 70.3 & 73.2 \\
    Ref (LLM reranker) & 0 & 64 & 2,530 & \underline{79.8} & \underline{89.1} \\
    \midrule
    \textbf{\method{}} & \textbf{0} & \textbf{0} & 2,540 & \textbf{80.7} & \textbf{89.4} \\
    \multicolumn{4}{l}{\quad\textit{gain over the strongest memory or selector}} & \gain{+10.4} & \gain{+12.8} \\
    \bottomrule
  \end{tabular}\hfill\setlength{\tabcolsep}{4.4pt}\begin{tabular}[t]{lrrrcc}
    \toprule
    \multicolumn{6}{c}{\textbf{LongMemEval-M} (500 sessions, 1.5M tokens per history)} \\
    \midrule
    & \multicolumn{2}{c}{LLM calls $\downarrow$} & & \multicolumn{2}{c}{Accuracy (\%) $\uparrow$} \\
    \cmidrule(lr){2-3}\cmidrule(lr){5-6}
    System & Write & Sel. & Tokens & Qwen & DeepSeek \\
    \midrule
    Adaptive-k \citep{taguchi2025adaptivek} & 0 & 0 & 1,180 & 57.8 & 60.8 \\
    DynamicRAG-8B \citep{sun2025dynamicrag} & 0 & 1 & 2,150 & 59.3 & 62.1 \\
    LongMemEval design \citep{wu2025longmemeval} & 503 & 0 & 2,530 & 56.3 & 60.3 \\
    HippoRAG~2 \citep{gutierrez2025hipporag2} & 1,025 & 1 & 2,520 & 46.1 & 48.1 \\
    Ref (LLM reranker) & 0 & 64 & 2,550 & \underline{66.6} & \underline{76.2} \\
    \midrule
    \textbf{\method{}} & \textbf{0} & \textbf{0} & 2,550 & \textbf{70.0} & \textbf{81.3} \\
    \multicolumn{4}{l}{\quad\textit{gain over the strongest memory or selector}} & \gain{+10.6} & \gain{+19.3} \\
    \bottomrule
  \end{tabular}
\end{table*}

\paragraph{Benchmarks and splits.}
Like most memory papers, we evaluate on LongMemEval-S \citep{wu2025longmemeval} and LoCoMo \citep{maharana2024locomo}.
LongMemEval-S (cleaned release) has 500 questions over histories of about 50 sessions and 100k tokens.
LongMemEval-M shares its questions over histories of about 500 sessions and 1.5M tokens, and serves as a harder third setting.
We exclude the 30 questions of the unanswerable type and compute $\kappa$ over the other 470.
A split fixed before any method run assigns 152 of the 500 questions of both variants to development and 348 to test, of which 328 are answerable.
The primary stratum is $\kappa\ge2$ on test, with $n=209$, and we also report $\kappa\ge3$ (49 test questions).
Label-only and oracle measurements have no learned component and use all 300 questions with $\kappa\ge2$.
LoCoMo has 1,540 scored questions in 10 conversations.
As the field does, we report all of them, and we add paired comparisons on test conversations 3 to 9, with conversations 0 to 2 for development.

\paragraph{Answering and judging.}
On LongMemEval, two answering models read every slate: Qwen3-32B \citep{yang2025qwen3} with FP8 weights under vLLM, and DeepSeek~V4.1~Flash \citep{deepseekai2026v41flash} through one pinned FP8 provider.
Both decode at temperature 0 with reasoning disabled and the official step-by-step answer prompt of LongMemEval.
Final labels come from the official GPT-4o judge with its per-type prompts.
A Qwen3-32B judge with the same prompts agrees with GPT-4o on 94\% to 98\% of answers and changes one primary label.
On LoCoMo, we use the answering model, prompt and judge of T-Mem: GPT-4o-mini answers, and the GPT-4o-mini judge is run three times per answer.
Appendix~\ref{app:protocol} gives the serving and judging details.

\paragraph{Implementation details.}
Single decodes are unstable: changing only the lowest-ranked 100 to 300 tokens of a 2,700-token slate changes the judged answer on 91 of 300 multi-session questions for Qwen and 43 for DeepSeek.
The variation comes from the answering model, because re-judging identical answers agrees on 99.9\% of them.
We therefore build $\Kj=4$ slates per question and system, with budgets $B\in\{2{,}400,2{,}500,2{,}600,2{,}700\}$ in the system's own order, and report the mean judge label.
Identical prompts are answered once, so Adaptive-k and DynamicRAG, whose slates often stop short of the budget, reuse one decode for all four budgets on 86\% and 52\% of questions. T-Mem assembles one context per question and uses one decode.
Systems are compared by a paired percentile bootstrap over questions with 10,000 draws at the 5\% level.
Labels are fixed before each run, and a difference is a \emph{gain} when it is significantly above zero and a \emph{loss} when it is significantly below zero.
It is at \emph{parity} when it is not significant and is significantly above $-3$ points in a one-sided test at the 2.5\% level, and \emph{underpowered} otherwise.
Parity is therefore non-inferiority at a 3-point margin with no significant difference.
In Table~\ref{tab:s}, every margin of \method{} over a memory or context-selection method is a gain, and every comparison with Ref is parity except a gain on LongMemEval-M with DeepSeek.
The 3-point margin is the spread of one system's accuracy across four near-identical budgets on one backend, 2.6 and 3.3 points.
On LongMemEval-M, this averaged protocol replicates a single-decode test (Appendix~\ref{app:lme-m}).
All systems in the headline tables are answered on one backend per answering model.

\paragraph{Compared systems.}
We compare against three families at the shared budget: context selectors, LLM-built memories and an LLM reranker.
The context selectors are Adaptive-k \citep{taguchi2025adaptivek}, which applies its official largest-gap rule to bge-m3 similarities over the whole history, and DynamicRAG \citep{sun2025dynamicrag}, whose 8B selector chooses among the top 40 fused turns.
The LLM-built memories are the LongMemEval memory design \citep{wu2025longmemeval}, HippoRAG~2 \citep{gutierrez2025hipporag2}, A-Mem \citep{xu2025amem} and T-Mem \citep{guo2026tmem}.
They run with Qwen3-32B as the builder LLM, so the comparison is between memory designs under a shared builder.
The LongMemEval design was run on LongMemEval-M only.
T-Mem, the strongest system we found on LoCoMo, is the main comparison, reproduced on LoCoMo and built by us on LongMemEval-S.
The LLM reranker (\emph{Ref}) scores each turn of the 64-turn fused pool from 0 to 10 with Qwen3-8B and fills the budget in score order.
Appendix~\ref{app:baselines} lists every deviation from the official settings, and the appendix documents all implementation details, including seed selection.

\section{Results}
\label{sec:results}

The experiments test the recommender formulation through three research questions.
RQ1 asks whether selecting memories the way a recommender selects items outperforms memory systems that build structure with LLM calls.
RQ2 asks which terms of the stage factorization carry the gain, and RQ3 asks what the formulation costs.

\subsection{RQ1: Recommendation against LLM-built memory}
\label{sec:res-s}

On LongMemEval-S, \method{} is the most accurate system for both answering models, and it outperforms Adaptive-k, DynamicRAG, HippoRAG~2, A-Mem and T-Mem by 10.4 to 31.2 points, all gains (Table~\ref{tab:s}, left).
The margins are largest over HippoRAG~2 and A-Mem, whose slates hold all evidence for 0.50 of the question and budget pairs, against 0.94 for \method{}, and the factorization attributes most of the margins to this difference.
A recognition-filter prompt adapted to Qwen lowers the share of questions on which HippoRAG~2 keeps no fact from 36.8\% to 23.9\% and leaves its accuracy unchanged ($+0.4$ and $+1.1$), and \method{} still leads by 20.5 and 25.1 points.
T-Mem, the strongest LLM-built memory, calls an LLM about 2,270 times per history and trails by 10.4 and 16.1 points.
At its context length of 3,310 tokens, \method{} still leads by 8.6 and 15.4 points, both gains, so the gap is not explained by the shorter slate of \method{}.

Against Ref, which makes 64 LLM calls per query, \method{} is at parity for both answering models ($+1.0$ and $+0.2$).
At $\kappa\ge3$ ($n=49$), the margins over Adaptive-k, DynamicRAG, HippoRAG~2 and A-Mem grow to 23.0 to 44.9 points, and \method{} is the most accurate system for both answering models, with 68.4 and 77.6 (Figure~\ref{fig:cost}b).
At $\kappa=2$, its margins over every memory and context-selection method are 6.1 to 27.0 points, all gains.
From $\kappa=2$ to $\kappa\ge3$, the two rankers lose 15 to 17 points, whereas the memory and context-selection methods lose 17 to 40.
Scattered evidence penalizes selectors that miss a session, since a slate without a gold turn loses 37 to 64 points of accuracy with Qwen.
Over all 328 answerable test questions, it exceeds Adaptive-k, DynamicRAG and HippoRAG~2 by 7.9 to 17.8 points and stays at parity with Ref.
These margins come from a ranker trained on development labels, and a label-free ranker distilled from LLM scores also exceeds every memory and context-selection method.
The margins do not rest on sessions shared with training.
On the 126 test questions whose gold sessions appear in no development pool, \method{} exceeds Adaptive-k, DynamicRAG, HippoRAG~2 and A-Mem by 8.3 to 25.2 points, all gains.

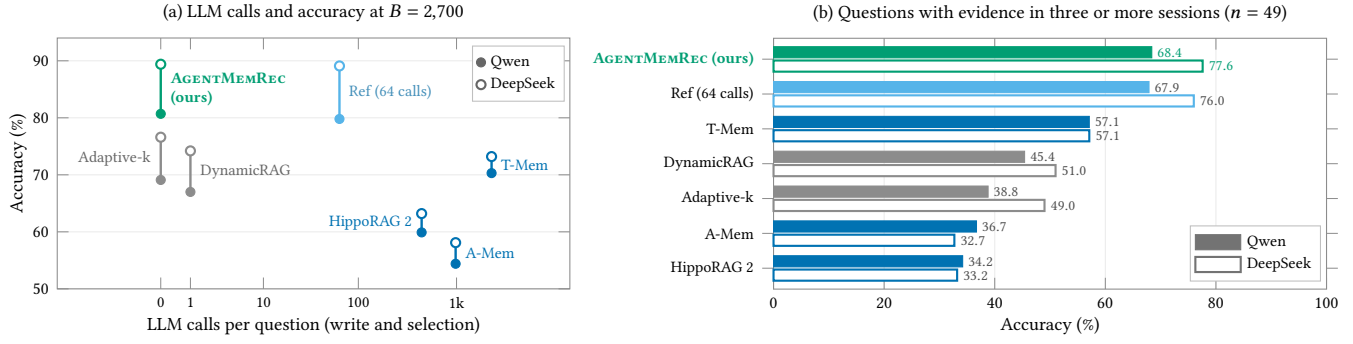
\begin{figure*}[t]
  \centering
  \pgfplotsset{
    costax/.style={
      height=4.9cm, scale only axis=false,
      tick label style={font=\scriptsize}, label style={font=\footnotesize}, title style={font=\footnotesize, yshift=-3pt},
      ymajorgrids, grid style={black!8}, axis line style={black!45}, tick style={black!45},
      ylabel={Accuracy (\%)}, ylabel style={yshift=-4pt}, clip=false,
    },
    sysq/.style={only marks, mark=*, mark size=1.7pt, #1, fill=#1},
    sysd/.style={only marks, mark=*, mark size=1.7pt, #1, fill=white, thick},
    seg/.style={#1, line width=0.9pt},
  }
  \tikzset{lab/.style={font=\scriptsize, inner sep=1pt, #1}}
  \begin{tikzpicture}
    \begin{axis}[costax, width=0.47\textwidth,
        title={(a) LLM calls and accuracy at $B=2{,}700$},
        xmin=-1.05, xmax=4.15, ymin=50, ymax=94,
        xtick={0,0.30103,1.04139,2.00432,3.00043}, xticklabels={0,1,10,100,1k},
        xlabel={LLM calls per question (write and selection)}, xlabel style={yshift=3pt},
        legend style={at={(0.99,0.97)}, anchor=north east, font=\scriptsize, draw=black!25, inner sep=1.5pt, row sep=-1pt},
        legend cell align=left]
      \addlegendimage{sysq=black!60}\addlegendentry{Qwen}
      \addlegendimage{sysd=black!60}\addlegendentry{DeepSeek}
      \addplot[seg=cOurs, forget plot] coordinates {(0,80.7) (0,89.4)};
      \addplot[sysq=cOurs, forget plot] coordinates {(0,80.7)};
      \addplot[sysd=cOurs, forget plot] coordinates {(0,89.4)};
      \node[lab=cOurs, anchor=west, font=\scriptsize\bfseries, align=left] at (axis cs:0.07,85.05) {\method{}\\(ours)};
      \addplot[seg=black!45, forget plot] coordinates {(0,69.1) (0,76.6)};
      \addplot[sysq=black!45, forget plot] coordinates {(0,69.1)};
      \addplot[sysd=black!45, forget plot] coordinates {(0,76.6)};
      \node[lab=black!60, anchor=east] at (axis cs:-0.07,72.85) {Adaptive-k};
      \addplot[seg=black!45, forget plot] coordinates {(0.30103,67.0) (0.30103,74.2)};
      \addplot[sysq=black!45, forget plot] coordinates {(0.30103,67.0)};
      \addplot[sysd=black!45, forget plot] coordinates {(0.30103,74.2)};
      \node[lab=black!60, anchor=west] at (axis cs:0.37,70.6) {DynamicRAG};
      \addplot[seg=cLLM, forget plot] coordinates {(1.81291,79.8) (1.81291,89.1)};
      \addplot[sysq=cLLM, forget plot] coordinates {(1.81291,79.8)};
      \addplot[sysd=cLLM, forget plot] coordinates {(1.81291,89.1)};
      \node[lab=cLLM!80!black, anchor=west] at (axis cs:1.88,84.45) {Ref (64 calls)};
      \addplot[seg=cMem, forget plot] coordinates {(2.64640,59.9) (2.64640,63.2)};
      \addplot[sysq=cMem, forget plot] coordinates {(2.64640,59.9)};
      \addplot[sysd=cMem, forget plot] coordinates {(2.64640,63.2)};
      \node[lab=cMem, anchor=east] at (axis cs:2.58,61.55) {HippoRAG~2};
      \addplot[seg=cMem, forget plot] coordinates {(2.99123,54.4) (2.99123,58.1)};
      \addplot[sysq=cMem, forget plot] coordinates {(2.99123,54.4)};
      \addplot[sysd=cMem, forget plot] coordinates {(2.99123,58.1)};
      \node[lab=cMem, anchor=west] at (axis cs:3.06,56.25) {A-Mem};
      \addplot[seg=cMem, forget plot] coordinates {(3.35584,70.3) (3.35584,73.2)};
      \addplot[sysq=cMem, forget plot] coordinates {(3.35584,70.3)};
      \addplot[sysd=cMem, forget plot] coordinates {(3.35584,73.2)};
      \node[lab=cMem, anchor=west] at (axis cs:3.42,71.75) {T-Mem};
    \end{axis}
  \end{tikzpicture}\hfill
  \begin{tikzpicture}
    \begin{axis}[costax, width=0.5\textwidth, xbar, ylabel={},
        title={(b) Questions with evidence in three or more sessions ($n=49$)},
        xmin=0, xmax=100, ymin=0.4, ymax=7.6, ytick={1,...,7},
        yticklabels={HippoRAG~2, A-Mem, Adaptive-k, DynamicRAG, T-Mem, Ref (64 calls), {\color{cOurs}\bfseries\method{} (ours)}},
        y tick label style={font=\scriptsize}, ymajorgrids=false, xmajorgrids,
        xlabel={Accuracy (\%)}, xlabel style={yshift=3pt},
        nodes near coords={\pgfmathprintnumber[fixed, fixed zerofill, precision=1]{\pgfplotspointmeta}}, every node near coord/.append style={font=\tiny, black!70, xshift=-1pt},
        nodes near coords align={horizontal},
        legend style={at={(0.99,0.03)}, anchor=south east, font=\scriptsize, draw=black!25, inner sep=1.5pt, row sep=-1pt},
        legend cell align=left]
      \addlegendimage{area legend, fill=black!55, draw=black!55}\addlegendentry{Qwen}
      \addlegendimage{area legend, fill=white, draw=black!55, thick}\addlegendentry{DeepSeek}
      \addplot[bar width=4.2pt, bar shift=2.6pt, fill=black!45, draw=black!45, forget plot] coordinates {(38.8,3) (45.4,4)};
      \addplot[bar width=4.2pt, bar shift=-2.6pt, fill=white, draw=black!45, thick, forget plot] coordinates {(49.0,3) (51.0,4)};
      \addplot[bar width=4.2pt, bar shift=2.6pt, fill=cMem, draw=cMem, forget plot] coordinates {(34.2,1) (36.7,2) (57.1,5)};
      \addplot[bar width=4.2pt, bar shift=-2.6pt, fill=white, draw=cMem, thick, forget plot] coordinates {(33.2,1) (32.7,2) (57.1,5)};
      \addplot[bar width=4.2pt, bar shift=2.6pt, fill=cLLM, draw=cLLM, forget plot] coordinates {(67.9,6)};
      \addplot[bar width=4.2pt, bar shift=-2.6pt, fill=white, draw=cLLM, thick, forget plot] coordinates {(76.0,6)};
      \addplot[bar width=4.2pt, bar shift=2.6pt, fill=cOurs, draw=cOurs, forget plot, every node near coord/.append style={cOurs, font=\tiny\bfseries}] coordinates {(68.4,7)};
      \addplot[bar width=4.2pt, bar shift=-2.6pt, fill=white, draw=cOurs, thick, forget plot, every node near coord/.append style={cOurs, font=\tiny\bfseries}] coordinates {(77.6,7)};
    \end{axis}
  \end{tikzpicture}
  \Description{Two panels of accuracy on LongMemEval-S. Panel a plots accuracy against LLM calls per question: AgentMemRec at zero calls is as accurate as the reranker with 64 calls, and the LLM-built memories with hundreds to thousands of calls are lower. Panel b shows accuracy on questions whose evidence spans three or more sessions, where AgentMemRec has the highest accuracy for both answering models.}
  \caption{Accuracy on the LongMemEval-S test questions with evidence in at least two sessions, for Qwen (filled) and DeepSeek (hollow). (a) Every compared system by its LLM calls per question at write and selection time ($n=209$). (b) The questions whose evidence spans three or more sessions. \method{} reaches the accuracy of the 64-call reranker with no LLM call, and it is the most accurate system on the questions with the most scattered evidence.}
  \label{fig:cost}
\end{figure*}

\paragraph{LongMemEval-M\@.}
\label{sec:res-m}
On the 1.5M-token histories of LongMemEval-M, \method{} is again the most accurate system for both answering models.
It exceeds Adaptive-k, DynamicRAG, the LongMemEval memory design and HippoRAG~2 by 10.6 to 33.3 points, all gains.
It is at parity with Ref for Qwen ($+3.3$) and exceeds it for DeepSeek ($+5.1$).
This margin over Ref is larger than on LongMemEval-S, consistent with the wider pool, which raises slate completeness on LongMemEval-M by 3.8 points.

\paragraph{Shared protocols.}
\label{sec:res-locomo}
Table~\ref{tab:locomo} places \method{} among memory systems under the protocols that memory papers share, with GPT-4o-mini as the answering model.
On LongMemEval-S, with all 500 questions and a GPT-4o judge, \method{} scores 82.5 on the 348 test questions and 83.2 on all 500, including the development questions on which its ranker was trained.
Both lie above the results that MemOS, TiMem, RecMem and Zep report, which is indicative because judges and data releases can differ.
On LoCoMo, with its write-time records and relative-date tags at T-Mem's context of about 3,800 tokens, \method{} reaches parity with our T-Mem reproduction, 77.0 against 77.4.
Its ranker is the one trained on LongMemEval, applied unchanged, and a ranker trained on LoCoMo development conversations performs at parity with it on test (Appendix~\ref{app:extra}).
It scores above Mem0, Zep and MemOS, and above our reproduction of HyperMem from its official code, which scores 72.6 with the same reader ($+4.4$, a gain).
Its records cost \$0.07 against \$16 for the memory of our T-Mem reproduction.
For RQ1, the recommender formulation outperforms every memory and context-selection method on LongMemEval-S and LongMemEval-M and matches or exceeds the 64-call reranker there, and on LoCoMo it is at parity with T-Mem.

\subsection{RQ2: Which terms of the factorization carry the gain}
\label{sec:rq3}
Table~\ref{tab:designs} splits each system's accuracy into slate completeness and its accuracies on complete and incomplete slates.

\paragraph{Completeness carries the margins.}
Slate completeness ranges from 0.50 for HippoRAG~2 and A-Mem to 0.94 for \method{}, and across its 11 designs it orders accuracy with a Spearman correlation of 0.95 for both answering models.
This supports the product of pool and ranking completeness as a design objective.
A symmetric split of the factorization assigns to completeness 12.0 to 22.1 points of the 11.6 to 26.3-point margins over the four turn-level methods with Qwen, and 12.3 to 23.1 of 12.8 to 31.2 with DeepSeek.
Accuracy on incomplete slates adds 2.4 to 8.4 points, because the incomplete slates of \method{} still hold 63\% of the gold turns, against 34\% to 45\% for these methods, and accuracy on complete slates changes the margins by $-4.7$ to $0.0$.
On question and budget pairs where both slates are complete, \method{} and each compared system differ by $-2.8$ to $+2.2$ points, none significantly.
HippoRAG~2 and A-Mem, whose slates are rendered turns, thus answer about as well as \method{} once the evidence is in the slate, and their LLM calls leave half of their slates incomplete.
T-Mem assembles scenes and items that do not map to turns, so the factorization omits it.
For memory design, this places the bottleneck in the selection stages that set pool and ranking completeness.

\paragraph{Candidate generation.}
The fused 128-turn pool holds all evidence for 0.981 of the questions, against 0.938 for the 64-turn pool of Ref.
Ref completes more of what its pool holds, and the product still favours \method{}, 0.940 against 0.915, as the bound $\rslate\le\rpool$ predicts.
On Ref's own pool, the ranker of \method{} reaches Ref's completeness (0.920 against 0.915), so its advantage comes from the wider pool that a budget of no LLM calls affords.
Widening the pool from 64 to 128 turns raises slate completeness on LongMemEval-M from 0.823 to 0.861.

\paragraph{Learning to rank.}
The LambdaRank model raises ranking completeness from 0.783 for the recall order to 0.959, and from 0.420 to 0.867 at $\kappa\ge3$.
It adds 10.0 and 10.9 points over the recall order and 6.7 and 9.0 over the cross-encoder score, all gains, and the symmetric completeness term accounts for 8.3 and 8.6 of the first two.
In a test question with evidence in five sessions, it moves the five evidence turns from recall ranks 1 to 67 to ranks 1 to 7, and only its slate keeps all of them (Table~\ref{tab:case}).
In another, the single evidence turn names neither the competition nor the purchase, and the ranker lifts it from recall rank 89 to rank 14, the only selection that answers correctly.
Retraining locates this gain in the combination of relevance and turn context: the relevance signals alone reach 0.815, the four context features alone 0.537, and dropping any single context feature costs at most 3 points.
Without any turn-context feature, the ranker still exceeds Adaptive-k, DynamicRAG, HippoRAG~2 and A-Mem by 6.2 to 24.6 points, all gains.
A LightGBM ranker distilled from Qwen3-8B scores, with no gold labels, reaches 0.933, is at parity with it ($+1.2$ and $+1.0$) and exceeds every memory and context-selection method by 9.2 to 30.3 points.
LLM rankers that score the same pool with 128 calls per question reach 0.943 and 0.944 and differ from \method{} by $-1.8$ to $+1.6$ points, none significantly.
A ranker trained on LoCoMo labels alone trails the in-domain ranker by 5.5 and 6.3 points and still exceeds every memory and context-selection method, except T-Mem with Qwen ($+4.9$, not significant).

\paragraph{Slate construction.}
At half and at twice the budget, \method{} differs from Ref by at most 1.8 points, none significantly.
Even with half the budget, it stays above the accuracy that every memory and context-selection method reaches at 2,700 tokens.

\paragraph{Optional write-time records.}
Without the extension and at 2,700 tokens, the cascade scores 75.2 on LoCoMo, above Mem0, Zep and our HyperMem reproduction with no LLM call at write time. Its temporal errors come mostly from complete slates with mis-resolved relative dates, the term $\acomp$ that the tags target.
The first enrichment gains 1.9 points on LoCoMo test conversations, and the final one changes LongMemEval-M accuracy by $0.0$ and $-0.4$ points (Appendix~\ref{sec:ext}).
The tags raise LoCoMo temporal accuracy by 7.2 points and lower LongMemEval-M temporal reasoning for DeepSeek by 8.2 (Qwen judge) through accuracy on complete slates, so their value depends on the benchmark.
For RQ2, slate completeness carries the gains, candidate generation and learning to rank raise it, and accuracy on jointly complete slates shows no significant difference between systems.

\subsection{RQ3: Cost}
\label{sec:efficiency}

We count a generation by a decoder language model as an LLM call, so \method{} makes none at selection time, and its cross-encoder pass is part of its measured serving cost.
On LongMemEval-M, its median selection latency is 329\,ms against 955\,ms for Ref, and on LongMemEval-S its 90th-percentile latency is 321\,ms against 1,156\,ms (Table~\ref{tab:serving}).
At saturation on one A100, selection takes 280 GPU-milliseconds per request for \method{} against 891 for Ref, 3.2 times less, or \$69 against \$220 per million requests at a list price of \$0.89 per A100-hour.
Adding the shared pool construction to both keeps the ratio at 3.0 in GPU time and, with a persistent per-history index, 2.7 in median latency at LongMemEval-M size.
At write time, the default pipeline makes no LLM call, whereas HippoRAG~2, A-Mem and T-Mem make about 440 to 2,270 per history, a cost that a deployed memory spreads over the requests of each user.
On LongMemEval-M, the LongMemEval design and HippoRAG~2 make 251,622 and about 512,000 write-time calls in total.
Figure~\ref{fig:cost}a places every system by its LLM calls per question at write and selection time.
\method{} is the only system that reaches the accuracy of the 64-call reranker with no LLM call, and the LLM-built memories spend hundreds to thousands of calls for 10 to 31 points less.
For RQ3, \method{} costs 3.2 times less than the 64-call reranker per request at the same accuracy, and it spends nothing on LLM calls at write time.

\section{When Item-Wise Ranking Suffices}
\label{sec:analysis}\label{sec:transfer}

The reduction to item-wise ranking and the advantage of the fill over cutoffs rest on two conditions on the reward.
Accuracy on complete and on incomplete slates must not depend on their composition, and adding a unit must not lower the value of a slate.
This section tests both conditions on the evaluated slates and on label-controlled slates.

\begin{table}[t]
  \centering\small
  \caption{Accuracy (\%) under the shared protocols of memory papers, with GPT-4o-mini answering. \method{} scores above every reported system on LongMemEval-S and is at parity with our T-Mem reproduction on LoCoMo. $^\ast$Copied (Appendix~\ref{app:baselines}), reprod.\ our reproduction, and $^\dagger$with the optional write-time records. Best in bold, second underlined.}
  \label{tab:locomo}
  \setlength{\tabcolsep}{2.6pt}
  \begin{tabular}[t]{>{\raggedright\arraybackslash}p{2.44cm}r}
    \toprule
    \multicolumn{2}{c}{\textbf{LongMemEval-S}} \\
    \midrule
    System & Acc.\ $\uparrow$ \\
    \midrule
    Zep$^\ast$ \citep{rasmussen2025zep} & 63.8 \\
    RecMem$^\ast$ \citep{dai2026recmem} & 69.2 \\
    TiMem$^\ast$ \citep{li2026timem} & 76.9 \\
    MemOS$^\ast$ \citep{li2025memos} & \underline{77.8} \\
    \midrule
    \textbf{\method{}} & \textbf{83.2} \\
    \quad test only & (82.5) \\
    \bottomrule
  \end{tabular}\hfill
  \begin{tabular}[t]{>{\raggedright\arraybackslash}p{3.27cm}r}
    \toprule
    \multicolumn{2}{c}{\textbf{LoCoMo}} \\
    \midrule
    System & Acc.\ $\uparrow$ \\
    \midrule
    Zep$^\ast$ \citep{rasmussen2025zep} & 59.2 \\
    Mem0$^\ast$ \citep{chhikara2025mem0} & 64.6 \\
    MemOS$^\ast$ \citep{li2025memos} & 75.8 \\
    T-Mem$^\ast$ \citep{guo2026tmem} & \textbf{80.3} \\
    Full context & 69.9 \\
    HyperMem, reprod.\ \citep{yue2026hypermem} & 72.6 \\
    T-Mem, reprod. & \underline{77.4} \\
    \midrule
    \textbf{\method{}}$^\dagger$ & 77.0 \\
    \bottomrule
  \end{tabular}
\end{table}

\paragraph{Complementarity.}
Every system loses 37 to 64 points of accuracy with Qwen when its slate misses a gold turn, and 37 to 68 with DeepSeek, and this gap $\acomp-\ainc$ scales each stage gain of the factorization.
Completeness falls with the number of evidence sessions, from 0.940 at $\kappa\ge2$ to 0.832 at $\kappa\ge3$ for \method{}.
A slate thus succeeds when it holds all of the evidence, which is the regime in which the item-wise reduction makes completeness the objective.

\paragraph{Composition.}
Across the compared systems, accuracy on jointly complete slates differs by at most 2.8 points at $\kappa\ge2$.
Accuracy on incomplete slates varies more, from 22\% to 47\%, and tracks how much of the gold an incomplete slate still holds.
The first condition thus holds for complete slates and only approximately for incomplete ones, and it begins to fail at $\kappa\ge3$.
There, the jointly complete slates of \method{} score 8.7 points below those of the cross-encoder order for Qwen and 6.5 points below for DeepSeek.
They hold 23 units against 11, most of them turns that the cross-encoder ranks low.
Label-controlled slates that fill the budget around the gold turns show the effect directly (LongMemEval-M, one decode, Qwen judge).
Ranked near-miss filler (Appendix~\ref{app:gaps}) costs 17.7 points for Qwen and 4.7 points for DeepSeek against the gold turns alone, and 39.4 and 14.1 points at $\kappa\ge3$.
Random filler of the same length costs 7.0 points for Qwen and nothing measurable for DeepSeek ($-2.3$).
Adding units can therefore lower the reward, by an amount that depends on the answering model.
These costs do not show that the reward is non-additive, since negative item values would produce them too.

\paragraph{Monotonicity along rankings.}
The fill beats every cutoff only if the reward does not fall as the fill adds units to a prefix of the ranking.
\begin{table}[t]
  \centering\small
  \caption{Stage factorization on LongMemEval-S test, $\kappa\ge2$ ($n=209$), accuracies in \%, so $\mathrm{Acc}=\ainc+(\acomp-\ainc)\,\rslate$. Slate completeness separates the systems, and \method{} has the most complete slates with no LLM call. Lower rows rank its 128-turn pool ($\rpool=0.981$, against 0.938 for the 64-turn pool of Ref). Best in bold, second underlined.}
  \label{tab:designs}
  \setlength{\tabcolsep}{4.8pt}
  \begin{tabular}{lccccc}
    \toprule
    & & \multicolumn{2}{c}{Qwen} & \multicolumn{2}{c}{Acc.\ $\uparrow$} \\
    \cmidrule(lr){3-4}\cmidrule(lr){5-6}
    System & $\rslate$ $\uparrow$ & $\acomp$ & $\ainc$ & Qwen & DeepSeek \\
    \midrule
    Adaptive-k & 0.684 & 87.1 & 30.3 & 69.1 & 76.6 \\
    DynamicRAG & 0.630 & 86.3 & 34.0 & 67.0 & 74.2 \\
    HippoRAG~2 & 0.502 & 89.5 & 30.0 & 59.9 & 63.2 \\
    A-Mem & 0.502 & 86.2 & 22.4 & 54.4 & 58.1 \\
    Ref (64 calls) & 0.915 & 83.8 & 36.6 & 79.8 & 89.1 \\
    \midrule
    Recall order & 0.768 & 84.4 & 25.3 & 70.7 & 78.5 \\
    Cross-encoder & 0.756 & 86.2 & 36.3 & 74.0 & 80.4 \\
    LightGBM, distilled & 0.915 & 83.8 & 33.8 & 79.5 & 88.4 \\
    LLM 8B (128 calls) & \underline{0.926} & 84.4 & 46.8 & \underline{81.6} & \textbf{90.1} \\
    LLM 32B (128 calls) & 0.925 & 86.3 & 36.5 & \textbf{82.5} & 87.8 \\
    \midrule
    \textbf{\method{}} (0 calls) & \textbf{0.940} & 83.0 & 46.0 & 80.7 & \underline{89.4} \\
    \bottomrule
  \end{tabular}
\end{table}
On the ranking of \method{} on LongMemEval-S, choosing per question among 1,350, 2,700 and 5,400 tokens gains $+2.4$ and $0.0$ points over three near-identical 2,700-token slates, neither significant.
Along this ranking, and along LLM-ranked slates of LongMemEval-M, the reward does not fall beyond decode noise, so no per-question cutoff gains measurably.
The per-question cutoffs of Adaptive-k and DynamicRAG leave 32\% and 37\% of their slates incomplete.
The conditions of the item-wise reduction thus hold where the stages of \method{} operate and fail at $\kappa\ge3$ and with near-miss filler.
Set-level objectives that model the composition of a slate are the next recommender designs to test there.

\section{Conclusion}
\label{sec:conclusion}

This paper asked how an LLM agent should choose, for every request, the stored memories that enter its bounded context.
We showed that this choice is budgeted slate recommendation and placed memory systems in one design space of recommender cascades.
Its factorization reduces memory design to raising the completeness of each stage within an LLM budget.
Following this principle, \method{} selects memories without any LLM call.
On LongMemEval, it outperforms LLM-built memories and context selectors, and it matches or exceeds an LLM reranker at a fraction of the cost. On LoCoMo, it outperforms production memory layers.
Its gains come from slates that more often hold all of the evidence, which suggests that recommender-system principles offer an effective and efficient foundation for agent memory.
Set-level objectives, models of user state and learning from use are natural next steps.

\bibliographystyle{ACM-Reference-Format}
\bibliography{refs_compact}

\appendix
\section{Data, Rendering and Judging}
\label{app:protocol}

\paragraph{LongMemEval data.}
We use the cleaned LongMemEval release of 2025-09, so our numbers are not comparable with those of the original paper.
LongMemEval-M and LongMemEval-S hold about 4,894 and 490 turns per question and share their 500 question identifiers.
The frozen split has 51, 69, 22 and 10 development questions and 119, 160, 49 and 20 test questions with $\kappa=1$, $\kappa=2$, $\kappa\ge3$ and unanswerable.
Every answerable question has a gold turn.
\paragraph{LoCoMo data.}
LoCoMo has 10 conversations of 369 to 689 turns and 1,540 scored questions (332 with $\kappa\ge2$), 385 in development conversations 0 to 2 and 1,155 in test conversations 3 to 9.

\paragraph{Rendering and budget.}
On LongMemEval, selected turns are grouped by session, sessions are ordered by date, and each session block follows the official template with the session date and the turns as JSON.
Token counts use the o200k tokenizer, and, unlike the official turn-level retriever, we group the selected turns of one session into one chunk for every system.
On LoCoMo, the slate starts with a fixed header and lists the selected turns under session headers with dates, as speaker and text with image captions.

\paragraph{Answering and judging.}
Answers are limited to 600 tokens, and Qwen3-32B runs with FP8 weights under vLLM with thinking disabled.
DeepSeek~V4.1~Flash runs through OpenRouter with one pinned FP8 provider, no fallback and reasoning disabled.
Every system of the LongMemEval comparisons was answered on DeepInfra.
The GPT-4o judge uses the official per-type prompts at temperature 0, on the test split only, with duplicate judge prompts scored once.

\paragraph{Decoding variability.}
Re-answering identical slates on a second machine agrees per question on 85.6\% for Qwen and 95.7\% for DeepSeek, while re-judging identical answers agrees on 99.9\% (Cohen's $\kappa$ 0.997).
These observations agree with reports that hosted LLM inference is not deterministic at temperature 0 \citep{atil2025nondeterminism}, that batch size and hardware change outputs \citep{yuan2025nondeterminism}, and that prompt formatting changes accuracy \citep{sclar2024quantifying}.
Averaging over four perturbed slates changed no primary label.

\paragraph{Noise across budgets.}
\label{app:repro}
On one backend, the accuracy of the LLM-ranked reference across four near-identical budgets of 2,400 to 2,700 tokens on LongMemEval-M spans 2.6 points for Qwen and 3.3 for DeepSeek, the noise scale behind the 3-point margin.

\paragraph{Qwen judge calibration.}
On 2,000 LongMemEval-S answers by GPT-4o-mini, the Qwen3-32B judge agrees with GPT-4o on 97.3\% (Cohen's kappa 0.904).
\label{app:qwenjudge}
Under the Qwen judge, every primary label of both LongMemEval variants is the same as under GPT-4o except \method{} against Ref with DeepSeek on LongMemEval-M, which is parity ($+3.6$) instead of a gain ($+5.1$).

\section{Implementation Details}
\label{app:impl}
\begin{table}[t]
  \centering\small
  \caption{Feature-group ablation of the ranker (label-only, LongMemEval-S test, each variant retrained). The learned combination exceeds every single feature group by 13 to 41 points of slate completeness at $\kappa\ge2$.}
  \label{tab:ablation}
  \setlength{\tabcolsep}{5pt}
  \begin{tabular}{@{}lcc@{}}
    \toprule
    Features & $\rslate$, $\kappa\ge2$ & $\rslate$, $\kappa\ge3$ \\
    \midrule
    \textbf{All ten} & \textbf{0.949} & \textbf{0.852} \\
    without turn context & 0.815 & 0.541 \\
    \midrule
    retrieval signals only & 0.813 & 0.541 \\
    cross-encoder only & 0.676 & 0.352 \\
    key overlap only & 0.644 & 0.219 \\
    turn context only & 0.537 & 0.582 \\
    \bottomrule
  \end{tabular}
\end{table}

\paragraph{Selection.}
RRF fuses the BM25 and bge-m3 rankings with $k=60$.
The ranker is a LightGBM LambdaRank model with 31 leaves, learning rate 0.05, 300 trees, at least 20 samples per leaf and a fixed seed, trained on the 304 development pools of both variants with one query group per pool.
The cross-encoder feature is the sigmoid bge-reranker-v2-m3 score, and on 128-turn pools the pool-position feature is the global RRF rank.
The distilled ranker uses the same settings with a regression objective on the Qwen3-8B scores, and the code package gives every remaining setting.

\paragraph{LLM-ranked reference.}
Qwen3-8B scores each pooled turn from 0 to 10 with one call per turn, from a prompt that shows the question, its date, the session date, the role and the turn.
Turns are ordered by score, with ties broken by fused rank.

\paragraph{Write-time records.}
\label{sec:ext}
Qwen3-32B writes records once per session with a type, text, date, entities and cited turns, and tags count against the budget.
Without the records, the cascade trails T-Mem on LoCoMo by 2.3 points, and with them, at T-Mem's context of 3,800 tokens (the $^\dagger$ row of Table~\ref{tab:locomo}), it reaches parity.
On the 88 temporal-reasoning questions of LongMemEval-M, the tags change accuracy through complete slates ($-7.5$ of $-8.2$ points) and displace only 5 of 712 gold turn slots.

\begin{table}[t]
  \centering\small
  \caption{Accuracy (\%) on the LongMemEval-S test questions by the number of evidence sessions. \method{} is the most accurate system at $\kappa\ge3$, and at $\kappa=2$ it exceeds every memory and context-selection method by 6.1 to 27.0 points, all gains. Best in bold, second underlined.}
  \label{tab:kappa}
  \setlength{\tabcolsep}{4.5pt}
  \begin{tabular}{@{}lcccc@{}}
    \toprule
    & \multicolumn{2}{c}{$\kappa=2$ ($n=160$)} & \multicolumn{2}{c}{$\kappa\ge3$ ($n=49$)} \\
    \cmidrule(lr){2-3}\cmidrule(lr){4-5}
    System & Qwen $\uparrow$ & DeepSeek $\uparrow$ & Qwen $\uparrow$ & DeepSeek $\uparrow$ \\
    \midrule
    Adaptive-k & 78.4 & 85.0 & 38.8 & 49.0 \\
    DynamicRAG-8B & 73.6 & 81.2 & 45.4 & 51.0 \\
    HippoRAG~2 & 67.8 & 72.3 & 34.2 & 33.2 \\
    A-Mem & 59.8 & 65.9 & 36.7 & 32.7 \\
    T-Mem & 74.4 & 78.1 & 57.1 & 57.1 \\
    Ref (LLM reranker) & \underline{83.4} & \textbf{93.1} & \underline{67.9} & \underline{76.0} \\
    \midrule
    \textbf{\method{}} & \textbf{84.5} & \underline{93.0} & \textbf{68.4} & \textbf{77.6} \\
    \bottomrule
  \end{tabular}
\end{table}

\paragraph{Additional results.}
\label{app:extra}
On LongMemEval-S, the paired differences of \method{} minus each system, in points for Qwen and DeepSeek, are $+11.6$ and $+12.8$ for Adaptive-k, $+13.8$ and $+15.2$ for DynamicRAG, $+20.8$ and $+26.2$ for HippoRAG~2, $+26.3$ and $+31.2$ for A-Mem and $+10.4$ and $+16.1$ for T-Mem, all gains.
Table~\ref{tab:kappa} splits these questions by the number of evidence sessions, and the margins over Adaptive-k, DynamicRAG, HippoRAG~2 and A-Mem grow from 6.1 to 27.0 points at $\kappa=2$ to 23.0 to 44.9 points at $\kappa\ge3$.
In the feature ablation (Table~\ref{tab:ablation}), dropping the turn-context features costs 13.4 points of slate completeness at $\kappa\ge2$ and 31.1 at $\kappa\ge3$, while dropping role, position, age or length alone costs at most 2.9.
The pool position with the four turn-context features alone already reaches 0.903.
Training on 128-turn instead of 64-turn pools changes slate completeness by at most 0.1 points.
Shared sessions place 43 of the 490 test gold turns in development pools, and 74 of their 92 occurrences in training carry label 0.
On the 119 single-session test questions, \method{} scores 89.5 and 94.3, at parity with Ref, and on LoCoMo the ranker trained on LongMemEval performs at parity with one trained on LoCoMo development conversations.

\paragraph{Seeds, latency and cost.}
The LongMemEval split, the ranker, the random filler, DynamicRAG and the bootstrap each use one fixed seed, and the answering models decode at temperature 0.
Latency is the median over 50 sequential test questions on one A100 with warm servers, and serving cost runs all 209 test questions with Ref at 32 requests in flight and \method{} at 8.
Neither includes the shared BM25 and dense pool, which was timed separately for the S histories and for proxies of LongMemEval-M size, and dollars use \$0.89 per A100-hour.

\section{Case Studies}
\label{app:cases}
Table~\ref{tab:case} follows three test questions whose evidence the cheap rankings and the memory systems miss.

\paragraph{Ordering events.}
The five evidence turns of the concerts question open five sessions, and each mentions a concert in passing while asking for something else.
In the 128-turn pool, the recall order ranks them 1, 14, 38, 61 and 67 and the cross-encoder 11, 16, 23, 36 and 47, so their slates keep one and two of them, and Ref's 64-turn pool misses one.
The ranker of \method{} moves them to ranks 1, 2, 3, 6 and 7, in line with the weight of the turn-context features in Table~\ref{tab:ablation}, and only its slate keeps all five.

\begin{table}[t]
  \centering\small
  \caption{Three LongMemEval-S test questions. \emph{Concerts}: What is the order of the concerts and musical events I attended in the past two months, starting from the earliest? \emph{Festivals}: How many movie festivals that I attended? \emph{Investment}: I mentioned an investment for a competition four weeks ago? What did I buy? Ev.\ counts evidence turns in the 2,700-token slate, and OK counts correct slates per answering model (T-Mem answers once from its own memory). \method{} keeps all evidence and answers correctly in all three, and bold marks every evidence turn kept and every slate correct.}
  \label{tab:case}
  \setlength{\tabcolsep}{3pt}
  \begin{tabular}{@{}lcccccc@{}}
    \toprule
    & \multicolumn{2}{c}{Concerts} & \multicolumn{2}{c}{Festivals} & \multicolumn{2}{c}{Investment} \\
    \cmidrule(lr){2-3}\cmidrule(lr){4-5}\cmidrule(lr){6-7}
    System & Ev.$\uparrow$ & OK$\uparrow$ & Ev.$\uparrow$ & OK$\uparrow$ & Ev.$\uparrow$ & OK$\uparrow$ \\
    \midrule
    Adaptive-k & 2/5 & 0/4 & 3/6 & 0/4 & 0/1 & 0/4 \\
    DynamicRAG-8B & 2/5 & 0/4 & 5/6 & 0/4 & 0/1 & 0/4 \\
    HippoRAG~2 & 0/5 & 0/4 & 3/6 & 0/4 & 0/1 & 0/4 \\
    A-Mem & 1/5 & 0/4 & 4/6 & 2/4 & 0/1 & 0/4 \\
    T-Mem & own & 0/1 & own & 1/1 & own & 0/1 \\
    Ref (64 calls) & 3/5 & 0/4 & \textbf{6/6} & \textbf{4/4} & 0/1 & 0/4 \\
    \midrule
    \textbf{\method{}} & \textbf{5/5} & \textbf{4/4} & \textbf{6/6} & \textbf{4/4} & \textbf{1/1} & \textbf{4/4} \\
    \bottomrule
  \end{tabular}
\end{table}

\paragraph{Counting events.}
The festivals question asks for a count, which needs every mention, and its six evidence turns lie in three sessions.
The recall order ranks them 4 to 51, and \method{} ranks all six within the first 14 and keeps them in its slate, as Ref does with 64 LLM calls.
Every selector that drops a mention gives a wrong count, except A-Mem on two of four slates.

\begin{table*}[t]
  \centering\small
  \caption{LongMemEval-M test, $\kappa\ge2$ ($n=209$), with the GPT-4o judge and four perturbed slates per question. Accuracies are in \%. LLM calls are generations per question at selection time, or in total at write time, and tokens are approximate mean slate lengths. $\Delta$ is \method{} minus the row in points, with its label: G gain, P parity, ref.\ a reference arm without a label. Best accuracy in bold, second underlined.}
  \label{tab:main}
  \setlength{\tabcolsep}{4pt}
  \begin{tabular}{@{}llrrrcccc@{}}
    \toprule
    & & Write & Sel. & & \multicolumn{2}{c}{Accuracy $\uparrow$} & \multicolumn{2}{c}{$\Delta$} \\
    \cmidrule(lr){6-7}\cmidrule(lr){8-9}
    System & Pool / ranker & LLM calls & LLM calls & Tokens & Qwen & DeepSeek & Qwen & DeepSeek \\
    \midrule
    Adaptive-k \citep{taguchi2025adaptivek} & history / largest gap & 0 & 0 & 1,180 & 57.8 & 60.8 & $+12.2$ G & $+20.6$ G \\
    DynamicRAG-8B \citep{sun2025dynamicrag} & fused 40 / listwise 8B & 0 & 1 & 2,150 & 59.3 & 62.1 & $+10.6$ G & $+19.3$ G \\
    LongMemEval design \citep{wu2025longmemeval} & rounds / dense, fact keys & 251,622 & 0 & 2,530 & 56.3 & 60.3 & $+13.6$ G & $+21.1$ G \\
    \quad native top-10, uncapped & rounds / dense, fact keys & 251,622 & 0 & 4,945 & 64.1 & 65.6 & $+5.9$ ref. & $+15.8$ ref. \\
    HippoRAG~2 \citep{gutierrez2025hipporag2} & rounds / fact graph, PPR & $\approx$512k & 1 & 2,520 & 46.1 & 48.1 & $+23.9$ G & $+33.3$ G \\
    Ref & fused 64 / Qwen3-8B & 0 & 64 & 2,550 & \underline{66.6} & \underline{76.2} & $+3.3$ P & $+5.1$ G \\
    \midrule
    \textbf{\method{}} & fused 128 / LightGBM & 0 & 0 & 2,550 & \textbf{70.0} & \textbf{81.3} & \multicolumn{2}{c}{reference} \\
    \bottomrule
  \end{tabular}
\end{table*}

\paragraph{Resolving time.}
The investment question refers to an event four weeks before the question date.
Its one evidence turn, from that day, reads ``I actually got my own set of sculpting tools \ldots{} today'' without naming the competition.
The recall order ranks it 89 and the cross-encoder 65, Ref's 64-turn pool misses it, and even LLM rankers that score all 128 turns with Qwen3-8B or Qwen3-32B leave it out of their slates.
\method{} ranks it 14 and keeps it, and it is the only system that answers correctly.

\paragraph{Ordering three events.}
A fourth question asks for the order of three sports events in one month, each mentioned in the opening turn of a session.
The recall order ranks the three turns 17, 29 and 86, and Ref's 64-turn pool misses the last, so Ref answers wrongly on every slate.
\method{} ranks them 1, 2 and 10 and answers correctly on every slate, while Adaptive-k, DynamicRAG, HippoRAG~2 and A-Mem keep at most one of them.

\paragraph{What the cases share.}
In all four questions, the recall order ranks some evidence turns too low for the budget, and the slates of the memory and context-selection methods miss evidence.
The learned ranker lifts every evidence turn into the slate, the completeness term that carries the margins of Section~\ref{sec:rq3}.
It does so with no LLM call, where Ref spends 64 per question and still misses evidence in three of the four questions.

\section{Baselines and Deviations}
\label{app:baselines}
\paragraph{Adaptive-k.}
We follow the official code and settings: the rule keeps the turns up to the largest gap in the sorted bge-m3 similarities of the whole history, ignoring the extreme 10\% tail, and then five more turns.
The kept turns fill the slate in similarity order up to the budget cap, which binds for 4\% of questions.

\paragraph{DynamicRAG\@.}
We run the released DynamicRAG-8B with the official template and prompt on the top 40 fused turns, each with its session date and role and truncated to 600 characters in the selector input, sampling at temperature 0.4.
The selected turns are rendered in full in the returned order up to the budget, and every output was non-empty on the 209 LongMemEval-S test questions.

\paragraph{LongMemEval memory design.}
Values are rounds, keys are user turns expanded with user facts, and the ranking is dense.
The extractor is Qwen3-32B instead of Llama-3.1-8B, with one call per unique user message, and a bge-m3 arm was chosen on development data over the official Stella encoder.
Slates hold whole rounds at the shared budgets, where the paper uses ten uncapped rounds (run here as a reference arm), and time-aware query expansion was not run.

\begin{table}[t]
  \centering\small
  \caption{Selection-stage serving cost on one A100 80\,GB GPU for the 209 LongMemEval-S test questions with $\kappa\ge2$. \method{} makes no LLM call and needs about a third of the GPU time and latency of the 64-call reranker.}
  \label{tab:serving}
  \setlength{\tabcolsep}{4pt}
  \begin{tabular}{@{}lrrr@{}}
    \toprule
    & \method{} & Ref & Ratio \\
    \midrule
    LLM calls per request & \textbf{0} & 64 & n/a \\
    LLM tokens per request & \textbf{0} & 20,861 & n/a \\
    GPU time per request (ms) & \textbf{280} & 891 & 3.2$\times$ \\
    Dollars per million requests & \textbf{69} & 220 & 3.2$\times$ \\
    Median latency, S (ms) & \textbf{277} & 924 & 3.3$\times$ \\
    90th-percentile latency, S (ms) & \textbf{321} & 1,156 & 3.6$\times$ \\
    Median latency, M (ms) & \textbf{329} & 955 & 2.9$\times$ \\
    \midrule
    \multicolumn{4}{@{}l}{\textit{with the shared pool construction}} \\
    GPU time per request (ms) & \textbf{300} & 911 & 3.0$\times$ \\
    Median latency, M size (ms) & \textbf{375} & 1,001 & 2.7$\times$ \\
    \bottomrule
  \end{tabular}
\end{table}

\paragraph{HippoRAG~2.}
We use one passage per round and one index per history, with the official retrieval, recognition filter and personalized PageRank.
Qwen3-32B replaces Llama-3.3-70B, bge-m3 replaces NV-Embed-v2, and the filter prompt tuned for Llama is used verbatim. The filter keeps no fact for 189 of 500 LongMemEval-M and 36.8\% of the LongMemEval-S test questions, by well-formed choice. One added paragraph, fixed on development questions, tells the filter that I and my refer to the user and to keep partly relevant facts.

\paragraph{A-Mem.}
A-Mem runs on LongMemEval-S with one build per question by Qwen3-32B and a fix of one missing import in the official code.
Retrieval returns up to 50 notes, and we report retrieved turns, which led A-Mem's own notes for both answering models in a first run (54.1 against 44.0, and 58.4 against 50.6).

\paragraph{Copied results.}
The protocols of Table~\ref{tab:locomo} are those of the copied papers: all 500 LongMemEval-S questions, and the 1,540 scored LoCoMo questions with the GPT-4o-mini judge run three times.
The copied LongMemEval-S rows use GPT-4o-mini answers with GPT-4o (Zep), GPT-4o-mini (RecMem) or an unnamed judge (TiMem, MemOS).
The copied LoCoMo rows come from the results table of T-Mem \citep{guo2026tmem}, which takes Zep, Mem0 and MemOS from MemOS \citep{li2025memos}.
Our T-Mem reproduction on LoCoMo scores 2.8 points below its reported value, and a same-day re-run is at parity with it ($+0.4$).

\paragraph{HyperMem.}
We ran the official code with Qwen3-32B in place of its GPT-4.1-mini builder, its released retrieval setting without the reranker, and its own embedding model.
It scores 72.6 with the shared reader and 87.5 with its own GPT-4.1-mini chain-of-thought reader, which on a conversation built with GPT-4.1-mini gives 92.1 against the reported 92.7, so the reader explains most of that gap.

\paragraph{T-Mem.}
On LongMemEval-S, each history becomes one user-assistant conversation, and Qwen3-32B (FP8, thinking disabled) replaces GPT-4.1-mini as its builder.
Our answering models read the context that its own QA step reads, persona included.

\section{Results on LongMemEval-M}
\label{app:lme-m}\label{sec:headline}\label{app:gaps}

Under GPT-4o with four slates (Table~\ref{tab:main}), \method{} is at parity with Ref for Qwen ($+3.3$) and exceeds it for DeepSeek ($+5.1$), which the Qwen judge labels parity.
It exceeds the other systems by 10.6 to 33.3 points, all gains, and the LongMemEval design and HippoRAG~2 reach slate completeness 0.495 and 0.342 against 0.861.
At write time, these two make 251,622 and about 512,000 LLM calls on LongMemEval-M, and \method{} makes none.

\paragraph{Single-decode primary test.}
\label{app:k1}
The single-decode runs (Qwen judge) were the primary test, with the averaged protocol as their replication: \method{} scores 68.9 and 79.9 against 65.6 and 76.1 for Ref ($+3.3$ and $+3.8$), parity for both models.

Label-controlled slates on LongMemEval-M ($\kappa\ge2$, all 300 questions, Qwen3-32B judge, one decode, $B=2{,}700$) give slate completeness and Qwen and DeepSeek accuracy of 0.823, 0.690 and 0.780 for the LLM ranker, and 0.850, 0.727 and 0.807 for gold in pool then cross-encoder order. With all gold first, they are 0.773 and 0.883 with cross-encoder order, 0.720 and 0.857 with LLM-ranked filler, and 0.827 and 0.927 with random filler. The gold turns alone (238 tokens) give 0.897 and 0.903.
With the gold turns fixed, ranked near-miss filler thus costs more than random filler of the same length for both answering models.

\end{document}